\documentclass[12pt]{article}
\usepackage{geometry}           
\usepackage{graphicx}
\usepackage{amssymb,amsmath,amsthm,bbm}
\usepackage[round]{natbib}
\usepackage{color,soul}
\usepackage[dvipsnames]{xcolor}
\usepackage{cancel}
\usepackage{listings,enumerate,hyperref,authblk}
\usepackage[utf8]{inputenc}

\newtheorem{theorem}{Theorem}
\newtheorem{lemma}[theorem]{Lemma}

\newcommand{\rev}[1]{\textcolor{black}{#1}}
\newcommand{\greg}[1]{\textcolor{black}{#1}}
\newcommand{\alex}[1]{\textcolor{black}{#1}}
\newcommand{\olivier}[1]{\textcolor{black}{#1}}

\renewcommand{\P}{\mathbb P}

\title{Research Article: Computing Individual Risks based on Family History in Genetic Disease in the Presence of Competing Risks}
\author[1,2]{G. Nuel\thanks{corresponding author, gregory.nuel@math.cnrs.fr}}
\author[3,4]{A. Lefebvre}
\author[5]{O. Bouaziz}
\affil[1]{LPMA, UMR CNRS 7599, Paris, France}
\affil[2]{UPMC, Sorbonne universités, Paris, France}
\affil[3]{UPSud, Paris-Saclay, Orsay, France}
\affil[4]{Institut Curie, Paris, France}
\affil[5]{MAP5, UMR CNRS 8145, Paris, France}

\begin{document}
\maketitle

\begin{abstract}
When considering a genetic disease with variable age at onset (ex: diabetes,  familial amyloid neuropathy, cancers, etc.), computing the individual risk of the disease based on family history (FH) is of critical interest both for clinicians and patients. Such a risk is very challenging to compute because: 1) the genotype $X$ of the individual of interest is in general unknown; 2) the posterior distribution $\P(X | \mathrm{FH},T>t)$ changes with $t$ ($T$ is the age at disease onset for the targeted individual); 3) the competing risk of death is not negligible. 

In this work, we present a modeling of this problem using a Bayesian network mixed with (right-censored) survival outcomes where hazard rates only depend on the genotype of each individual. We explain how belief propagation can be used to obtain posterior distribution of genotypes given the FH, and how to obtain a time-dependent posterior hazard rate for any individual in the pedigree. Finally, we use this posterior hazard rate to compute individual risk, with or without the competing risk of death. 

Our method is illustrated using the Claus-Easton model for breast cancer (BC). This model assumes an autosomal dominant genetic risk factor such as non-carriers (genotype 00) have a BC hazard rate $\lambda_0(t)$ while carriers (genotypes 01, 10 and 11) have a (much greater) hazard rate $\lambda_1(t)$. Both hazard rates are assumed to be piecewise constant with known values (cuts at 20, 30, \ldots, 80 years). The competing risk of death is derived from the national French registry. 

\vspace{1em}
\noindent Keywords: piecewise constant hazard, Bayesian network, belief propagation, Hardy-Weinberg, Mendelian transmission.
\end{abstract}

\section{Introduction}

Complex diseases with variable age at onset typically have many interacting factors such as the age, lifestyle, environmental factors, treatments, genetic inherited components. The genetic component is generally composed of one or several genes including major genes for which a deleterious mutation rises significantly the risk of the disease and/or minor genes which participation in the disease is moderate by itself. 

The mode of inheritance can be monogenic if a mutation in a single gene is transmitted or polygenic if mutations in several genes are transmitted. As an example of a major gene in a complex disease, the BRCA1 gene is well known to be strongly correlated with ovarian and breast cancer since the 90s \citep{hall1990linkage,claus1994autosomal}. Carriers of a deleterious mutation in BRCA1 gene have a much higher risk to be affected with relative risks ranging from 20 to 80 but deleterious mutations in BRCA1 gene only explain 5 to 10 \% of the disease \citep{mehrgou2016importance} as many other implicated known or unknown genes exist along with sporadic cases (cases with no inherited component).
    
In other rare genetic diseases such as the Transthyretin-related Hereditary Amyloidosis
(THA), no sporadic cases are found and therefore the incidence is equal to zero among non-carriers and all affected individuals are necessarily carriers of a deleterious mutation \citep{plantebordeneuve2003gst,alarcon2009pel}.

The family history (FH) of such diseases is often the first tool for clinicians to detect a family of carriers of a deleterious mutation as any unusual accumulation of cases in relatives leads to suspect a deleterious allele in the family. With the appropriate model and computation, the FH can be used to better target the most appropriate individuals for a genetic testing and / or to identify high-risk individuals who require special attention (monitoring and/or treatments).

The first challenge to compute such a model comes from the fact that genotypes are mostly (if not totally) unobserved and that posterior carrier probability computations must sum over a large number of familial founders' genotypes configurations. Once such computations are carried out, deriving posterior individual disease risk is also a challenging task since the posterior carrier distribution changes over time and must be accounted for. Finally, for diseases with possibly late age at onset (\emph{e.g.} cancer), the competing risk of death is not negligible and must be accounted for.
 
%A competing risk situation occurs when an event (called a competing event) precludes the occurrence of the event of interest. This is typically the case in cancer studies where for great ages the risk of death is not negligible. Ignoring the risk of death would amount to assume that death cannot happen and as a result, measures such as the cumulative incidence (the probability of having a cancer before any time point) would be overestimated. Famous examples of such situations include studies on dementia where the patients are of a particularly advanced age and have a high risk of dying as in \citet{jacqmin2014} or \citet{wanneveich2016}, or studies on geriatric patients \citep[see for instance][]{berry2010}.

A competing risk situation occurs when an event (called a competing event) precludes the occurrence of the event of interest. This is typically the case for late-onset diseases as the risk of death is not negligible for advanced age. Ignoring the risk of death would amount to assume that death cannot happen and would therefore lead to overestimate the cumulative incidence (the probability of having the disease before any time point). Famous examples of such situations include dementia where the patients are of a particularly advanced age and have a high risk of dying as in \citet{jacqmin2014} or \citet{wanneveich2016}, or studies on geriatric patients \citep[see for instance][]{berry2010}.

\olivier{Classical familial risk models such as Claus-Easton \citep{claus1991genetic,easton1993genetic}, BOADICEA \citep{antoniou2004boadicea}}, \greg{or the BayesMendel models \citep[BRCAPRO, MMRpro, PancPRO and MelaPRO, see][]{chen2006prediction}} \olivier{do not take into account the competing event of death. As a result, it is likely that individual predictions will tend to be overestimated from these models \citep{de2012estimation}.} \greg{The main result of the present work is that we show how to derive individual risk predictions from the family history while taking into account the competing risk of death, which is a new contribution to the best of our knowledge.}

\greg{Another interesting point is that, unlike most similar publications, we here provide all the necessary details to integrate the likelihood over the unobserved genotypes and to compute posterior genotype distributions using Bayesian network and sum-product algorithms. One should not that these models and algorithms clearly are often used in the context of genetics \citep[see][for a few examples]{lauritzen1996graphical,o1998pedcheck,fishelson2002exact,lauritzen2003graphical,palin2011identity}, but rarely fully detailed \citep[see][for example]{chen2006prediction}.}

\greg{It should also be noted that the genetics community usually prefers to rely on simple \emph{peeling} algorithms rather than Bayesian network for pedigree computations but the two concepts are in fact totally equivalent, and the sum-product algorithm presented in this paper can indeed be seen as a simple Bayesian network based reformulation of the most general peeling-based algorithm developed so far \citep{totir2009efficient}.}

%\olivier{We also show how Bayesian networks can be used in this context in order to derive posterior distributions of the genotype given the family history. Expliquer les avantages des reseaux bayesiens. Sommes nous les seuls a les utiliser dans ce contexte? Si oui, il faut le dire et donner ses avantages par rapport aux autres approches.} \alex{While Bayesian Networks are used by several authors in the framework of genetic analysis such as gene mapping, linkage analysis and quantitative traits analysis (\citep{lauritzen1996graphical,o1998pedcheck,fishelson2002exact,lauritzen2003graphical,palin2011identity}), classical genetic models combining survival data in the framework of complex diseases with variable age at onset do not use them despite their great potential in this context in terms of flexibility and efficiency.}

%\alex{Je connais aussi encore mal la littérature sur l'utilisation de BN en génétique et j'espère ne pas raconter trop de bêtise. Pour les réf, j'ai repris les tiennes, Grégory. Je n'ai pas trop développé car c'est déjà développé dans la partie carrier risk}

%In the present work we present three main contributions and an illustrative application example.
\greg{The paper is organized as follows:} firstly, in Section~\ref{sec:bn} we introduce a formal generic Bayesian network model adaptable to any genetic disease with variable age at onset. Secondly, in Section~\ref{sec:pcarrier}, we provide in this context all the necessary details to carry belief propagation on this model, and express the marginal posterior carrier distribution using Bayesian network's potentials. Thirdly, in Section~\ref{sec:risk}, we give closed-form formulas for the posterior individual disease risk, and introduce a simple numerical algorithm allowing to take into account the competing risk of death. Finally, in Section~\ref{sec:results}, all the methods are illustrated with the Claus-Easton model \greg{for breast cancer using the disease model and the parameters of  \citet{claus1991genetic,easton1993genetic}.} \olivier{In particular, individual predictions derived by taking into account the competing risk of death or ignoring it are compared, which emphasizes the importance of properly taking into account competing risk of death in such models.}

\section{Materials and Methods}

In this section, we first introduce our model (Section~\ref{sec:bn}) as a Bayesian network. We next explain how to perform belief propagation in order to obtain posterior carrier distributions (Section~\ref{sec:pcarrier}). Finally, we provide all the details needed to derive disease risks predictions from these posterior distributions, including taking into account the competitive risk of death (Section~\ref{sec:risk}). 

\subsection{The Bayesian Network}\label{sec:bn}

We consider a total of $n$ (related) individuals. With $\mathcal{I}=\{1,\ldots,n\}$, we denote by $\mathcal{F}\subset\mathcal{I}$ the subset of the founders (i.e. individuals without ancestors in the pedigree) and we denote by $\mathcal{I}\setminus\mathcal{F}$ the set of non-founders (i.e. with ancestors in the pedigree). Let $\boldsymbol X = (X_1,\ldots,X_n) \in \{00,01,10,11\}^n$ be the genotypic distribution\footnote{For the sake of simplicity, we consider here a simple bi-allelic gene but multi-allelic genes can obviously be easily considered.} of the whole family, where $X_i$ denotes the genotype of Individual $i$. Let $\boldsymbol T = (T_1,\ldots,T_n) \in \mathbb{R}^n $ be the time vector representing
the age at diagnosis of all individuals. The joint distribution of $(\boldsymbol X, \boldsymbol T)$ is given by:
\begin{equation}\label{eq:bn}
\mathbb{P} \left(\boldsymbol X, \boldsymbol T \right) = \underbrace{\prod_{i \in \mathcal{F}} \mathbb{P}(X_i)  \prod_{i\in \mathcal{I}\setminus \mathcal{F}} \mathbb{P} \left (X_i |  X_{\text{pat}_i} , X_{\text{mat}_i} \right)}_\text{genetic part}
\times 
\underbrace{ \prod_{i \in \mathcal{I}} \P \left (T_i |  X_i \right) }_\text{survival part}
\end{equation}
which corresponds to the definition of a Bayesian Network (BN). See \citet{koller2009probabilistic} for more details. The genetic part of Eq.~(\ref{eq:bn}) only relies on the ``classical''  Mendelian assumption that the distribution of a non-founder genotype only depends on the parental genotypes. The survival part makes the strong assumption that all $T_i$ are conditionally independent given $X_i$. This assumption is clearly not true when considering any other familial effect on the disease (\emph{e.g.} polygenic effect, environmental exposure, etc.) which is often taken into account using a familial random effect (often called \emph{frailty} in the survival context). Such familial random effect is for example assumed to account for a polygenic effect in the BOADICEA model \citep{antoniou2002comprehensive,antoniou2004boadicea}. \rev{Note that for the sake of simplicity, the symbol ``$\mathbb{P}$'' corresponds through the whole paper either to a probability measure or to a density.}

The extension of the present model to frailty models such as BOADICEA is clearly possible and, in many ways, quite straightforward. However, for the sake of simplicity, we focus here on a simpler model and will briefly discuss the extension in the conclusion section. However, even with the strong assumption that $T_i$ only depends on $X_i$, since (the basically unobserved) $\boldsymbol X$ has a strong correlation structure within the pedigree, so does $\boldsymbol T$.

We can see on Fig.~\ref{fig:fam4} an example of a moderate size (hypothetical) family with a severe history of breast and ovarian cancer. This family has a total of $n=12$ individuals with $\mathcal{F}=\{1,2,3,4\}$ and $\mathcal I \setminus \mathcal F=\{5,6,7,8,9,10,11,12\}$. 
There is no inbreeding \rev{(mating between individuals with a common ancestor)} in this family but a mating loop \rev{(two families joined more than once by mating)} due to the two brothers of the first nuclear family having children with two sisters of the second nuclear family. Such looped pedigree can be tricky to represent and this explains why Individual~7 appears twice (with an identity link) in Fig.~\ref{fig:fam4}.

%\rev{For those more familial with Bayesian networks than with genetics, one should not that loops in pedigree are not cycles at that the underlying conditional dependence structure of the model remains a proper directed acyclic graph even in the presence of pedigree with loops.}

\olivier{One should note that loops in pedigree are not the same as cycles in the Bayesian networks framework in the sense that the underlying conditional dependence structure of the model remains a proper directed acyclic graph even in the presence of pedigree with loops.}

\begin{figure}
\begin{center}
\includegraphics[width=0.8\textwidth,trim=20 620 170 50,clip]{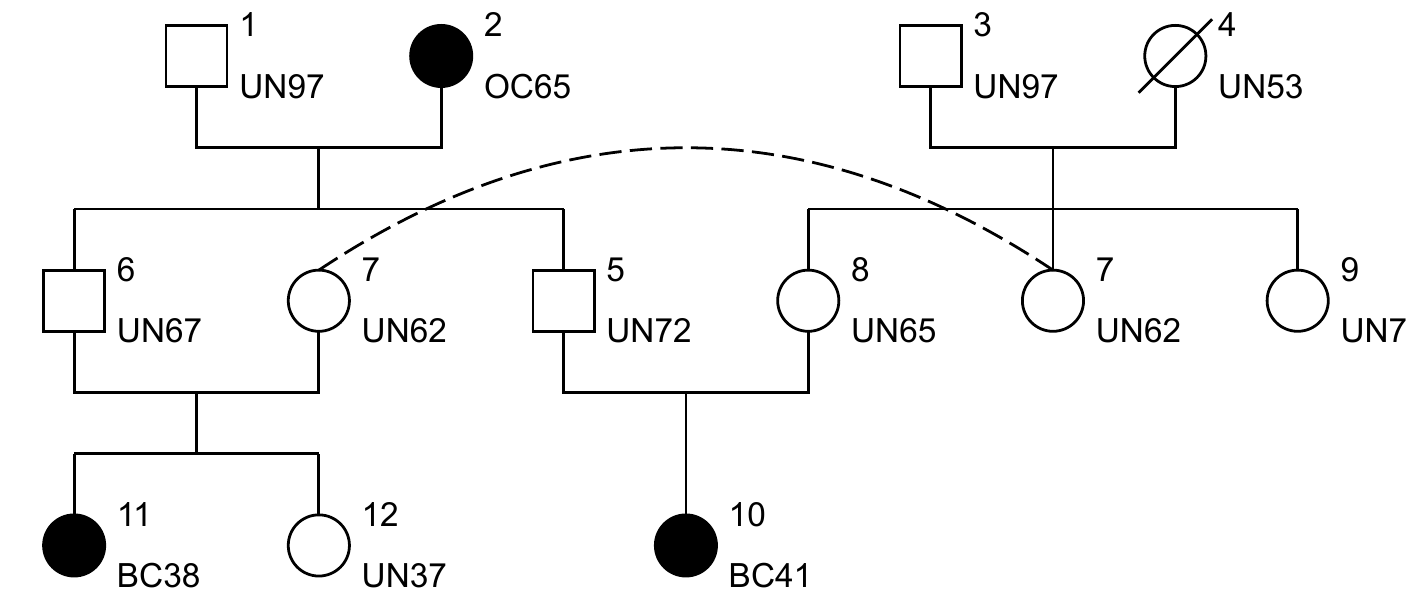}
\end{center}
\caption{A hypothetical family with a severe FH of cancer. \rev{Squares correspond to males, circles to females, and affected individual are filled in black.} Individual id on the top-right of the nodes, personal history of cancer (UN=UNaffected; BC=Breast Cancer; OC=Ovarian Cancer) on the bottom-right. The dashed line represents an identity link used to represent the mating loop \rev{(due to the mating between individuals 5/8 and 6/7)} between brothers 5 and 6, and sisters 7 and 8. }\label{fig:fam4}
\end{figure}

%%%%%%%%%%%%%%%%%%%%%%%%%%%%%%%%%%%%%

\paragraph{Genetic Part.}
 
 For the genetic part, we assume that founders' genotypes are distributed according to the 
Hardy-Weinberg distribution with disease allele frequency $f$. It means that for any founder $i\in\mathcal{F}$ we have $\mathbb{P}(X_i=00)=(1-f)^2$, $\mathbb{P}(X_i=01)=\mathbb{P}(X_i=10)=f(1-f)$, and $\mathbb{P}(X_i=11)=f^2$. \rev{This assumption is extremely frequent in family genetics and usually reasonable since it corresponds to the stationary distribution we observe in a population under mild assumptions. However, one should note that other distributions can easily be considered if necessary (\emph{e.g.} genotype $11$ forbidden because it is lethal).}
%\olivier{J'imagine que c'est: since it is the stationary distribution observed in a population under mild assumptions.}
 For the non-founder we simply assume a Mendelian transmission of the alleles, but unbalanced transmission patterns can also be considered. 

The genetic part of the model can also be easily extended to account for various constraints. For example, the presence of monozygous twins, say individuals $i$ and $j$, only requires one to add an identity variable between the two genotypes:  $I_{i,j}\in\{0,1\}$ such as $\mathbb{P}(I_{i,j} | X_i,X_j)=\mathbf{1}\{X_i=X_j\}$. Genetic tests (including error or not) can also be incorporated as additional variables $G_i$ such as $\P(G_i | X_i)$ corresponding to the test specificity and sensibility. Finally, assuming lethal genotypes (\emph{e.g.} genotype $11$) is done straightforwardly by setting to $0$ the probability of carrying such genotype.
\rev{This is equivalent to working conditionally on $\{X_i\neq 11 \text{ for all $i$}\}$ which obviously alter all genotype distributions, including Hardy-Weinberg for founders.}

%%%%%%%%%%%%%%%%%%%%%%%%%%%%%%%%%%%%%

\paragraph{Survival Part.}

We place ourselves in the classical survival framework, denoting by $\lambda(t)$  the (time dependent) hazard function, by $S(t)$ the survival function defined as $S(t) = \exp(-\Lambda(t))$ where $\Lambda(t)=\int_0^t \lambda(u)du$ is the cumulative hazard. 

We assume an autosomal dominant model where non-carriers have a disease incidence $\lambda_0(t)$ and carriers have a disease incidence $\lambda_1(t)$. This simple assumption results in the following expression of the survival part of the model:
\begin{equation}\label{eq:survpart}
%\mathbb{P}(T_i = t,\delta_i |  X_i )= 
%\left\{
%\begin{array}{ll}
%S_0(t) \lambda_0(t)^{\delta_i} & \text{if $X_i=00$} \\
%S_1(t) \lambda_1(t)^{\delta_i} & \text{if $X_i \neq 00$} \\
%\end{array}
%\right..
\left\{
\begin{array}{l}
\mathbb{P}(T_i > t |  X_i=00 )=S_0(t) \quad\text{and}\quad
\mathbb{P}(T_i = t |  X_i=00 )=S_0(t)\lambda_0(t)\\
\mathbb{P}(T_i > t |  X_i\neq 00 )=S_1(t)  \quad\text{and}\quad
\mathbb{P}(T_i = t |  X_i\neq 00 )=S_1(t)\lambda_1(t)
\end{array}
\right..
\end{equation}
\rev{As explained above, the symbol ``$\mathbb{P}$'' corresponds to a (conditional) probability measure for the event $\{T_i>t\}$ and to a density for the punctual event $\{T_i=t\}$}.

For example, in the context of the THA, non-carriers cannot be affected ($\lambda_0(t)\equiv 0$) and only carriers have an age-dependent incidence. In the context of breast cancer, $\lambda_0(t)$ might be the incidence for non BRCA carriers and $\lambda_1(t)$ the incidence for BRCA carriers (BRCA1 or BRCA2). 

Of course the simple model suggested in Eq.~(\ref{eq:survpart}) can easily be extended to account for other genetic models (\emph{e.g.} recessive, additive, gonosomal (\emph{i.e.} non-autosomal), with parent-of-origin effect, etc.) as well as for any known covariates (\emph{e.g.} BMI, smoking, other diseases, etc.) using a classical proportional hazard model. 
 
Hazard rates $\lambda_0(t)$ and $\lambda_1(t)$ are typically described by the literature as piecewise constant hazards (PCHs), but our model allows for any parametric or non-parametric shape as long as hazard rates are provided (\emph{e.g.} \rev{hazard rates of} Weibull distributions, Gaussian survival, etc.).

\subsection{Carrier Risk}\label{sec:pcarrier}

% Greg: Non je ne suis pas d'accord: ici on se limite au modèle qu'on suppose juste
% même sous ce modèle, même avec des tests parfaits, il est impossible d'observer X totalement
%
%It is essential to understand that X is at best partially observed as in most diseases, many implicated genes are still unknown. On the other hand, even in a situation with a known targeted gene and a genetic testing done ... etc

For all individual $i$ let us denote by $\mathrm{PH}_i$ his/her personal history of the disease. In the case where Individual $i$ was diagnosed with the disease at age $t_i$ we have $\mathrm{PH}_i=\{T_i=t_i\}$. If Individual $i$ was unaffected at age $t_i$ (age at the last follow-up), the variable $T_i$ is right-censored and we have $\mathrm{PH}_i=\{T_i>t_i\}$. From now on, we denote by $\mathrm{FH}$ the family history of the disease. This includes the personal history of all individuals and all possible additional constraints or informations (\emph{e.g.} monozygous twins, genetic tests, lethal alleles, etc.). \rev{Formally, we can define $\mathrm{FH}=\cup_i (\mathrm{PH}_i \cup \{X_i \in \mathcal{X}_i\})$ where $\mathcal{X}_i \subset \{00,01,10,11\}$ is the subset of allowed values for $X_i$ (\emph{e.g.} $\mathcal{X}_i = \{00,01,10\}$ if we know that the genotype $11$ is lethal, $\mathcal{X}_i = \{00\}$ if we know that a particular individual is a non-carrier, etc.).}
Even with genetic testing, it is essential to understand that $\boldsymbol X$ is, at best, partially observed. Indeed, even with a (hypothetical and unrealistic) 100\% specificity/sensitivity test, a positive heterozygous carrier status cannot distinguish between genotypes $01$ and $10$. Moreover, genetic tests are in general only available for a few individuals in the 
whole pedigree. Accounting for the unobserved genotypes is therefore of utmost importance.  

Following the classical BN notations, we write the so-called \emph{evidence} $\P(\mathrm{FH})$ as the simple following  sum-product of \emph{potentials}:
\begin{equation}\label{eq:sumproduct}
\P(\mathrm{FH}) = \sum_{X_1} \ldots \sum_{X_n} \prod_{i=1}^n K_i \left(X_i | X_{\mathrm{pa}_i} \right)
\end{equation}
where the potentials are defined by:
\begin{equation}
K_i \left(X_i | X_{\mathrm{pa}_i} \right) = \P(\mathrm{PH}_i | X_i) \times \left\{
\begin{array}{ll}
\P \left (X_i |  X_{\text{pat}_i} , X_{\text{mat}_i} \right) & \text{if $i \in \mathcal{I}\setminus \mathcal{F}$}\\
\P \left (X_i  \right) & \text{if $i \in  \mathcal{F}$}
\end{array}
\right.
\end{equation}
where $\P(\mathrm{PH}_i | X_i)$ is either $\P(T_i=t_i | X_i)$ or $\P(T_i>t_i | X_i)$ and can be obtained through Eq.~(\ref{eq:survpart}). \rev{Note that $\mathrm{pa}_i \subset \mathcal{I}$ denote the parental set of Individual $i$ (empty for founders), and that $X_\mathcal{J}=(X_j)_{j \in \mathcal{J}}$ %\alex{accolades non plutôt que parenthèses?}
for any $\mathcal{J} \subset \mathcal{I}$.}
As explained above, any additional information or constraint might and should be added directly into the potentials. 

Since $\boldsymbol X$ has $4^n$ possible configurations in the worst case, it is clearly impossible to simply enumerate these configurations even for moderate size pedigrees (e.g., for $n=10$ or $n=20$). We therefore need a more efficient algorithm to compute Eq.~(\ref{eq:sumproduct}). An efficient solution is provided by the Elston-Stewart algorithm \citep{elston1992elston} in the particular (and frequent) case where the pedigree has no loop. The basic idea is to eliminate variables from the sum-product (\emph{peeling} in the Elston-Stewart literature) from the last generations up to the oldest common ancestor. The resulting complexity $\mathcal{O}(n \times 4^3)$ clearly allows one to deal with arbitrary pedigree size as long as there is no loop. 

Unfortunately, loops (inbreeding or mating) are not totally uncommon in pedigrees and therefore have to be accounted for. A simple extension of the Elston-Stewart algorithm consists in using loop breakers: working conditionally to a few number of key genotypes that can be considered as duplicated individuals with known genotypes in a pedigree with no loop. For example, in Fig.~\ref{fig:fam4}, Individual $7$ is a possible loop breaker. By performing a classical Elston-Stewart algorithm for each genotypic configuration of the loop breakers, $\P(\mathrm{FH})$ can be computed with complexity $\mathcal{O}(n \times 4^{\ell+3})$ where $\ell$ is the number of loop breakers. 

In the context of Bayesian networks, computing $\P(\mathrm{FH})$ (and, in fact, the whole $\P(\boldsymbol X,\mathrm{FH})$ distribution) is typically done through \emph{belief propagation} (BP)\footnote{Also called \emph{sum-product} algorithm.} with a $\mathcal{O}(n \times 4^k)$ complexity where $k$ is the \emph{tree-width} of the graphical model \citep[see][for more details]{koller2009probabilistic}. For a pedigree with no loop, $k=3$ and the BP complexity is strictly the same than Elston-Stewart, but for more complex pedigrees, $k$ usually increases much slower than $\ell+3$ and, as a result, BP is often dramatically faster than Elston-Stewart with loop breakers. 

In order to achieve this, BP basically eliminates variables from the sum-product of Eq.~(\ref{eq:sumproduct}) in a suitable order. In that sense, it is very similar to the notion of \emph{cutset} long used to compute likelihoods in complex pedigrees \citep[see][for a recent reference on the MENDEL package]{lange2013mendel}. But BP has the noticeable advantage to allow obtaining the full posterior distribution $\P(\boldsymbol X |\mathrm{FH})$ for the same algorithmic complexity while likelihood-based approaches need to repeat many cutset eliminations to achieve the same results. As a consequence, it should not be surprising to see that, in parallel with the classical genetic literature \citep{elston1992elston,kruglyak1996parametric,lange2013mendel} many authors have been using BP and BN to deal with genetic models \citep{lauritzen1996graphical,o1998pedcheck,fishelson2002exact,lauritzen2003graphical,palin2011identity}.

\greg{Let us finally point out that the genetics community has put considerable efforts in developing Elston-Stewart algorithms for any Bayesian network counterpart, claiming that \emph{peeling-based} algorithms are more natural for geneticists than junction-tree based ones. Note however that the most general version of these \emph{peeling} algorithms \citep{totir2009efficient} is in fact \emph{exactly} equivalent to the classical junction-tree based forward/backward algorithm presented below.}

For completeness, we will now briefly recall all the minimal necessary results to implement BP in the context of our model. We nevertheless encourage the interested reader to refer to more classical references like \citet{lauritzen2003graphical} or \citet{koller2009probabilistic} for more details.

\paragraph{Variable Elimination and Junction Tree.}

\begin{figure}
\begin{center}
\includegraphics[width=1.0\textwidth,trim=40 30 90 30,clip]{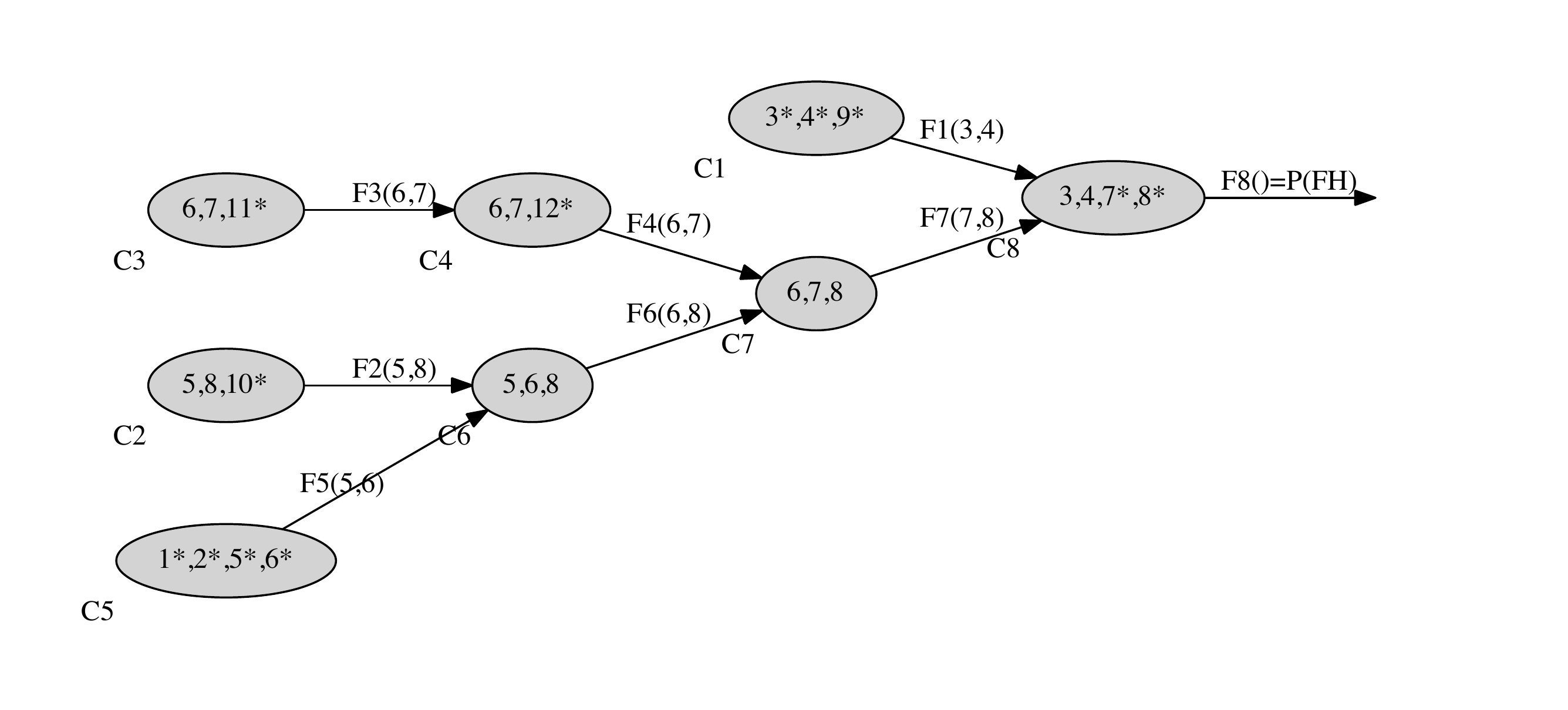}
\end{center}
\caption{Junction tree of our hypothetical family with the following elimination order:
$X_9,X_{10},X_{11},X_{12},X_{1,2},X_{5},X_6,X_{3,4,7,8}$.}\label{fig:jt}
\end{figure}

As an example, we consider the pedigree of Fig.~\ref{fig:fam4} and want to compute $\P(\mathrm{FH})$ by successive variable elimination. We use the following elimination order: $X_9$, $X_{10}$, $X_{11}$, $X_{12}$, $X_{1,2}$, $X_{5}$, $X_6$, and $X_{3,4,7,8}$. Here follow the quantities obtained in the process:
$$
F_1(X_{3,4})=\sum_{X_9} K_3(X_3) K_4(X_4) K_9(X_{3,4,9});
\  
F_2(X_{5,8})=\sum_{X_{10}} K_{10}(X_{5,8,10});
$$
$$
F_3(X_{6,7})=\sum_{X_{11}} K_{11}(X_{6,7,11});
\ 
F_4(X_{6,7})=\sum_{X_{12}} F_3(X_{6,7}) K_{12}(X_{6,7,12});
$$
$$
F_5(X_{5,6})=\sum_{X_{1,2}} K_1(X_1) K_2(X_2) K_5(X_{1,2,5})K_6(X_{1,2,6});
\ 
F_6(X_{6,8})=\sum_{X_5} F_2(X_{5,8}) F_5(X_{5,6});
$$
$$
F_7(X_{7,8})=\sum_{X_6} F_4(X_{6,7}) F_6(X_{6,8});
\ 
\P(\mathrm{FH})=\sum_{X_{3,4,7,8}} F_1(X_{3,4}) F_7(X_{7,8})
K_7(X_{3,4,7})K_8(X_{3,4,8}).
$$
We therefore can obtain $\P(\mathrm{FH})$ by considering only $6\times 4^3 + 2 \times 4^4=896$ configurations over the $4^{12}\simeq 16.8 \times 10^6$ total number of $\boldsymbol X$ configurations. Note that a memory bounded version of the variable elimination exists, see \citet{darwiche2001recursive} for more details.

Fig.~\ref{fig:jt} is a graphical representation of this particular sequence of elimination and is also a \emph{junction tree} defined as a set of $K$ \emph{cliques} $C_1,\ldots,C_K$ with $C_j \subset \{X_1,\ldots,X_n\}$ with the following properties:
\begin{enumerate}[i)]
\item tree: %for all $j<K$, $\{k>j, C_j \subset C_k\}\neq \emptyset$ and $\mathrm{to}_j = \min \{k>j, C_j \subset C_k\} \in \{j+1,\ldots,K\}$
each clique $j$ is connected to a a subsequent $clique$
$\mathrm{to}_j  \in \{j+1,\ldots,K\}$ ($\mathrm{to}_K=\text{root}$ by convention). We also define $\mathrm{from}_k=\{j, \mathrm{to}_j=k\}$ ($\mathrm{from}_1=\emptyset$) and $S_j=C_j \cap C_{\mathrm{to}_j}$ (with the convention that $S_K=\emptyset$).
\item covering: for all $i \in \{1,\ldots,n\}$ it exists a $j$ such as  $\{X_i,X_{\mathrm{pa}_i}\} \subset C_j$. We then define $\mathrm{of}_i=\min \{j, (X_i,X_{\mathrm{pa}_i}) \subset C_j\}$ and $C_j^*=\{X_i \in C_j, \mathrm{of}_i=j\}$.
\item running intersection: for all $i \in \{1,\ldots,n\}$ the subgraph formed by  $\{C_j, X_i \in C_j\}$ (and the from/to relationships) is a tree.
\end{enumerate}

In the graph theory, junction trees are used as an auxiliary structure for many applications (\emph{e.g.} graph coloring). The proof that any elimination sequence gives a junction tree can be found in \citet{koller2009probabilistic}. The \emph{tree-width} of an elimination sequence / junction tree is defined as the size of its largest clique. Finding the elimination sequence with the smallest tree-width is NP-hard in general, but many heuristics are available \citep{koller2009probabilistic}. The elimination order of Fig.~\ref{fig:jt} has been obtained using the well-known minimum fill-in heuristic.
%We present one of these heuristics, namely the \emph{min fill-in} heuristics, in Appendix~\ref{appendix:jt}.

\paragraph{Belief Propagation.}

We assume that a suitable elimination order / junction tree has been obtained. For all $j \in \{1,\ldots,K\}$ we hence define the potential of clique $C_j$ as
$\Phi_j(C_j)= \prod_{X_i \in C_j^*} K_i \left(X_i | X_{\mathrm{pa}_i} \right)$ and we have the following result:

\begin{theorem}{(posterior distribution)}\label{thm:post}
For all $i\in\{1,\ldots,n\}$, let $k=\mathrm{of}_i$ and we have:
$$
\P(X_i , \mathrm{FH}) = \sum_{C_k \setminus \{X_i\}} \left\{ 
 \prod_{j \in \mathrm{from}_k}  F_j(S_j) \times \Phi_k(C_k)  \times B_k(S_k)
 \right\}
$$
where the \emph{forward} quantities are defined for $k=1, \dots, K$ by:
$$
F_k(S_k) = \sum_{C_k \setminus S_k}  \left\{ 
 \prod_{j \in \mathrm{from}_k}  F_j(S_j) \times \Phi_k(C_k) 
\right\}
$$
and the \emph{backward} quantities are defined by $B_K(S_K=\emptyset)=1$ (convention) and for $k=K, \dots, 2$, for all $i \in \mathrm{from}_k$:
$$
B_i(S_i) = \sum_{C_k \setminus S_i} \left\{ 
 \prod_{j \in \mathrm{from}_k,j \neq i}  F_j(S_j) \times \Phi_k(C_k)  \times B_k(S_k)
\right\}.
$$
\end{theorem}
\begin{proof}
See Appendix~\ref{appendix:bp}.
\end{proof}

Using Theorem~\ref{thm:post}, it is therefore possible to obtain $\P(\mathrm{FH})$ and \emph{all} $\P(X_i | \mathrm{FH})$ by just recursively computing once all forward and backward quantities.

\subsection{Disease Risk}\label{sec:risk}

While the previous section covered the computation of the posterior probability $\P(X_i | \mathrm{FH})$ for all individuals in the pedigree, we now focus in this section on computing individual posterior disease risks, with or without the competing risk of death.

\paragraph{Risk without competing events.}

We consider an individual $i$ with a posterior carrier probability $\pi$ at age $\tau$, that is $\pi=\P(X_i \neq 00 | \mathrm{FH},{T}_i>\tau)$. Conditionally to the family history, we denote the survival and hazard functions respectively by $S$ and $\lambda$ such that, for $t\geq \tau$, $S(t)=\P({T}_i>t|\mathrm{FH},{T}_i>\tau)$ and $S(t) = \exp(-\int_{\tau}^t \lambda(u)du)$. We have the following result.

\begin{theorem}\label{th:survnocr}
%We consider an individual with a posterior carrier probability $\pi$ at age $\tau$. 
For any $t \geq \tau$, we have:
$$
S(t)= \pi \frac{S_1(t)}{S_1(\tau)} + (1-\pi) \frac{S_0(t)}{S_0(\tau)}
$$
$$
\P(X_i \neq 00 | \mathrm{FH}, {T}_i>t)= \frac{1}{S(t)} \pi \frac{S_1(t)}{S_1(\tau)}
$$
\begin{equation}\label{eq:posthaz}
\lambda(t) = \frac{1}{S(t)} \left[ \pi \frac{S_1(t)}{S_1(\tau)} \lambda_1(t) + (1-\pi) \frac{S_0(t)}{S_0(\tau)} \lambda_0(t) \right]
\end{equation}
\end{theorem}
\begin{proof}
See Appendix~\ref{appendix:risk}.
\end{proof}

\paragraph{Risk with death as a competing event.}

%PCH + $\alpha$ incidence BC and $\beta$ incidence BC+death. 

%A competing risk situation occurs when an event (called a competing event) precludes the occurence of the event of interest. This is typically the case in cancer studies where for great ages the risk of death is not negligible. Ignoring the risk of death would amount to assume that death cannot happen and as a result, measures such as the probability of having a cancer before any time point (called the cumulative incidence) would be overestimated. 
%On the other hand, 
As explained in the introduction, death precludes the occurence of the disease. This needs to be taken into account by defining 
%In order to take into account death as a competing event, 
the hazard rate of the disease conditionally to the fact that both disease and death have not occurred yet. From a statistical point of view, such a situation can be seen as a competing risk situation or as an illness-death model;  see \cite{ABGK} or \cite{andersen2012} for a presentation of such models. We define $T^*$ as the minimum between age at disease onset and age at death and we keep the notation ${T}$ to denote the age at disease onset. Given an individual $i$ with a family history $\mathrm{FH}$, its hazard rate for the disease is defined as
\[\lambda_{\alpha}(t)=\lim_{\Delta t\to 0}\frac{\P(t\leq {T}_i<t+\Delta t | T_i^*\geq t,\mathrm{FH})}{\Delta t}\]
We denote by $\lambda_\beta$ and $S_{\beta}$ the hazard and survival functions of $T_i^*$ (conditionally to the family history) and we assume that $\lambda_\alpha$ and $\lambda_\beta$ are piecewise constants with common cuts $\tau=c_0<c_1<\ldots<c_N$ (that is $\lambda_\alpha(t)=\alpha_j$ and $\lambda_\beta(t)=\beta_j$ for $t  \in ]c_{j-1},c_j]$). 

%could then be misleading as an individual with a low risk of cancer could just be the consequence 
%In contexts where different types of events can happen it is important to take 

\begin{lemma}\label{th:survcr} %if $\lambda_\alpha$ and $\lambda_\beta$ are PCH with common cuts $c_0<c_1<\ldots<c_N$ ($\lambda_\alpha(t)=\alpha_j$ for $t  \in ]c_{j-1},c_j]$) then 
For $j=1,\ldots,N$, $t  \in ]c_{j-1},c_j]$, we have
\rev{$$
\P(T_i\leq t| T_i > c_{j-1}, \mathrm{FH})=\int_{c_{j-1}}^{t} \lambda_\alpha(u) S_\beta(u) du = \frac{\alpha_j}{\beta_j} \left[ S_\beta(c_{j-1}) - S_\beta(t) \right]%\exp \left( -\Lambda_\beta(u)\right)
$$}
\end{lemma}
\begin{proof}
See Appendix~\ref{appendix:risk}.
\end{proof}

%\olivier{explain the cause-specific density in the competing risk context.}

\paragraph{Practical computations.}

We assume that one individual has a carrier probability $\pi$ at age $\tau$ (his age without the disease in the FH). We denote by $\lambda_\text{death}$ his/her hazard of death. Then the posterior disease risk with the competing risk of death can be computed through the following steps:
\begin{enumerate}[1)]
\item choose a fine enough discretization $\tau=c_0<c_1<\ldots<c_N=t_{\mathrm{max}}$ (ex: all $c_j-c_{j-1}=0.1$ year);
\item compute $\alpha_j=\lambda_\alpha(c_j)$ using Eq.~(\ref{eq:posthaz});
\item compute $\beta_j=\alpha_j + \lambda_\text{death}(c_j)$;
\item then the marginal posterior
probability of being diagnosed with the disease before age $c_k$, in the presence of death as a 
competing risk, is given for $k=1,\ldots,N$ by:
$$
\P(T_i \leqslant c_k | \mathrm{FH})=\sum_{j=1}^k \frac{\alpha_j}{\beta_j} \left[ S_\beta(c_{j-1}) - S_\beta(c_{j}) \right].
$$
\end{enumerate}

%\greg{proof of lemma + this result in appendix ?}

\section{Results and Discussion}\label{sec:results}

\subsection{The Claus-Easton Model}

In order to illustrate our method, we will use the model of illness and the parameters of the Claus-Easton model developed from the Cancer and Steroid Hormone Study in the 90s \citep{claus1991genetic,easton1993genetic}. 

The Claus-Easton model is a classical genetic model composed of a genotypic part and a phenotypic part with only the family history (FH) as covariate. It assumes an autosomal dominant mode of inheritance, and a piecewise constant hazard rate by steps of 10 years. The penetrance ($F(t)=1-S(t)$) and the density ($f(t)=\lambda(t)S(t)$) are given in Table~2 from \citet{easton1993genetic} for both carriers and non-carriers at ages $25, 35, \ldots, 85$. The hazard rates can therefore be derived from these data using the formula $\lambda(t)=f(t)/(1-F(t))$. The results of these computations are given in Table~\ref{tab:CEhaz}. The frequency of the mutated allele has been estimated at $f=0.0033$ \citep{claus1991genetic}. The death incidences needed in the competing risk section are given in Table~\ref{tab:deathhaz}.

Figure~\ref{fig:haz_surv} presents the incidence and survival for BC (carriers and non-carriers) as well as death. We can notice that the breast cancer incidences in carriers are always much higher than in non-carriers at any age and the relative risk between carriers and non-carriers is especially large ($\mathrm{RR}>50$) before age 40 (see Table~\ref{tab:CEhaz}) but then decreases with aging. We notice that the death incidence stays above the BC incidence for non-carriers at all ages and exceeds even the BC incidence for carriers from age 80. This shows the importance of taking it into consideration especially over a certain age.

%The death incidence stays below to the BC incidence before 40 but becomes higher from 40 in non-carriers and from 80 in carriers. We hence expect an important effect of the competing risk of death from these ages.

%The death incidence  is negligible in comparison with BC incidence before the 40s and become even higher than de carriers BC incidence after age 80. We hence expect an important effect of the competing risk with death from these ages. 

\begin{table}
$$
\begin{array}{crrrrrrr}
\hline
& 20-30 & 30-40 & 40-50 & 50-60 & 60-70 & 70-80 &
>80\\
\hline
\text{non carriers} & 2.00 & 26.04 & 112.94 & 139.94 & 235.17
& 232.16 & 232.03\\
\text{carriers} & 168.35 & 1391.49 & 3153.21 & 3222.22 &
3281.25 & 3289.86 & 3286.43\\
\text{relative risk} & 84.17 & 53.44 & 27.92 & 23.03 & 13.95 & 14.17 &
14.16 \\
\hline
\end{array}
$$
\caption{Annual incidence (for 100,000) of breast cancer (BC) for
carriers/non-carriers and relative risks by age (in
years) in the Claus-Easton model.}\label{tab:CEhaz}
\end{table}

\begin{table}
$$
\begin{array}{crrrrrrr}
\hline
20-30 & 30-40 & 40-50 & 50-60 & 60-70 & 70-80 & 80-85\\
23.85375 & 46.86641 & 130.5396 & 308.9539 & 599.914 & 1493.6 & 3845.406\\
\hline
85-90 & 90-95 & 95-99 & 99-100 & 100-101 & 101-102 & 102-103\\
8114.203 & 16400.99 & 27912.22 & 35644 & 38696.22 & 43033.07 & 45647.85\\
\hline
\end{array}
$$
\caption{Annual female death incidence (for 100,000) by age (in years) in the metropolitan French population between 2012 and 2014 \citep{ined}.
}\label{tab:deathhaz}
\end{table}

\begin{figure}
\begin{center}
\begin{tabular}{cc}
\includegraphics[width=0.47\textwidth,trim=0 0 0 0,clip]{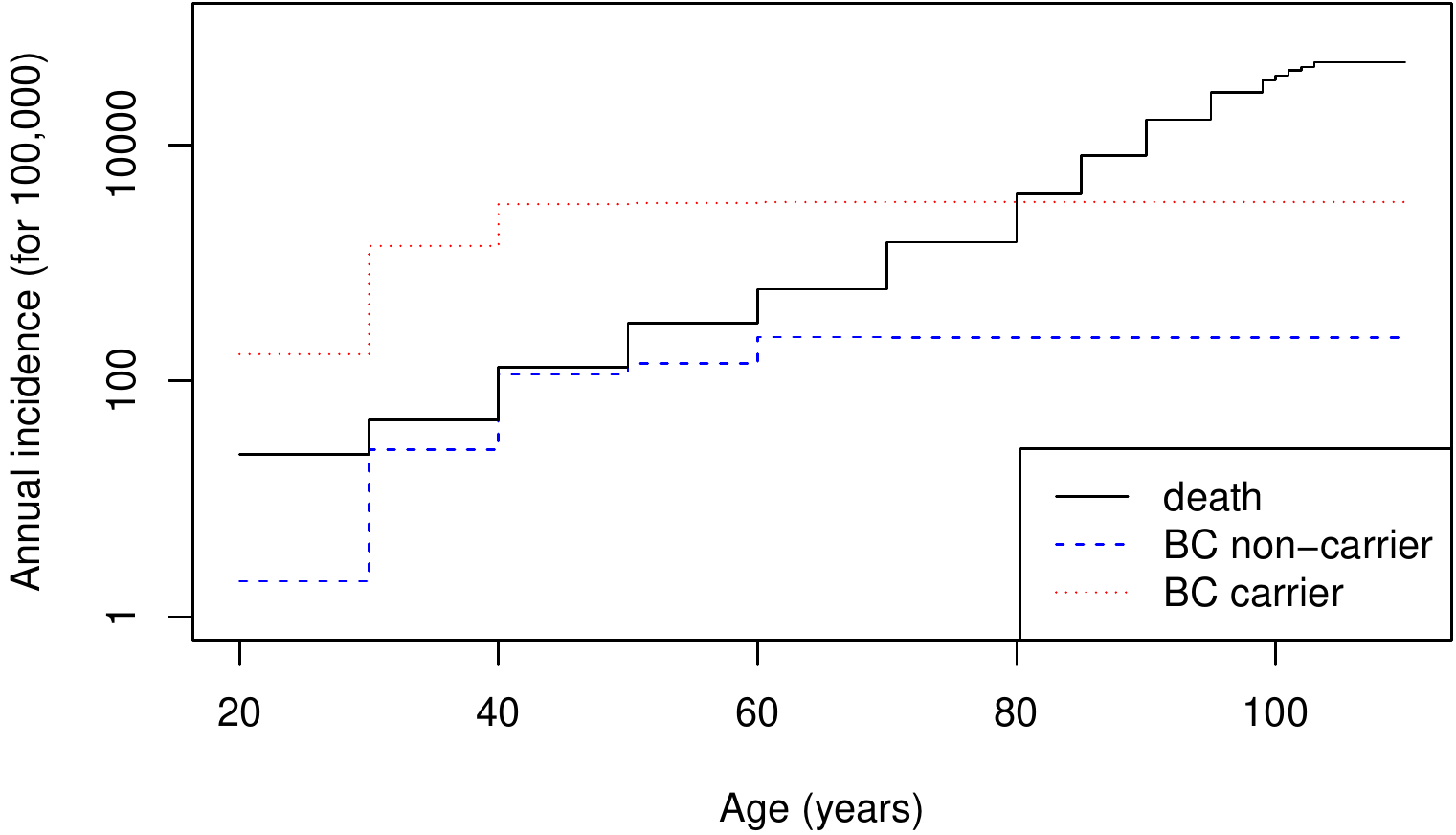}&
\includegraphics[width=0.47\textwidth,trim=0 0 0 0,clip]{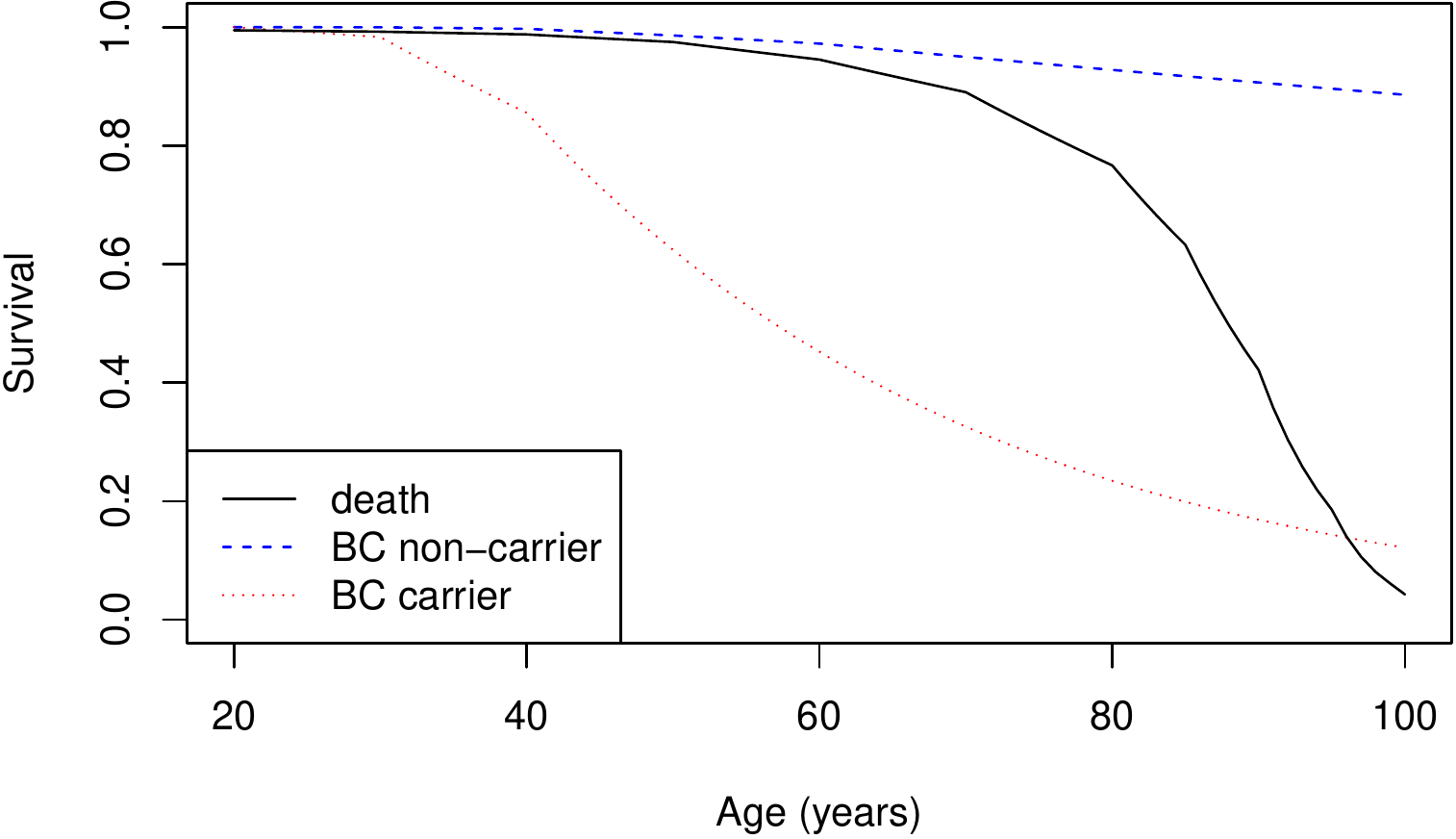}
\end{tabular}
\end{center}
\caption{Left-panel: annual (female) death incidence and annual non-carrier/carrier breast cancer incidence. Right-panel: death survival and percentage of non-carrier/carrier individuals without diagnosed breast cancer.}\label{fig:haz_surv}
\end{figure}

\subsection{Carrier Risk}

In this section we will use the belief propagation in Bayesian networks to obtain the posterior distribution of individual genotypes given the FH. We get the posterior probabilities of each genotype (non-carrier, heterozygous carrier with a paternal mutated allele, heterozygous carrier with a maternal mutated allele and homozygous carrier). 

%\begin{table}
%\centering
%\begin{tabular}{rrrrr}
%\hline
% Ind. & non-carrier & paternal carrier &  maternal carrier & homozygous carrier\\
% \hline
%1 & {\bf 96.758} & 1.614 & 1.614 & 0.014\\
%2 & 3.604 & {\bf 47.921} & {\bf 47.921} & 0.554\\
%3 & {\bf 99.445} & 0.277 & 0.277 & 0.001\\
%4 & {\bf 99.681} & 0.159 & 0.159 & 0.000\\
%5 & 10.010 & 1.562 & {\bf 87.410} & 1.019\\
%6 & 6.682 & 1.616 & {\bf 90.694} & 1.008\\
%7 & {\bf 99.466} & 0.338 & 0.194 & 0.001\\
%8 & {\bf 99.504} & 0.314 & 0.180 & 0.001\\
%9 & {\bf 99.778} & 0.140 & 0.081 & 0.000\\
%10 & 13.425 & {\bf 86.154} & 0.363 & 0.059\\
%11 & 8.271 & {\bf 91.266} & 0.401 & 0.062\\
%12 & {\bf 55.370} & 44.376 & 0.209 & 0.046\\
%\hline
%\end{tabular}
%\caption{Marginal posterior probabilities (in percent) for all individuals in our hypothetical family. The largest probabilities for each individual are highlighted in bold.}\label{tab:post}
%\end{table}

\begin{figure}
	\begin{center}
		\includegraphics[width=0.8\textwidth,trim=0 0 0 0,clip]{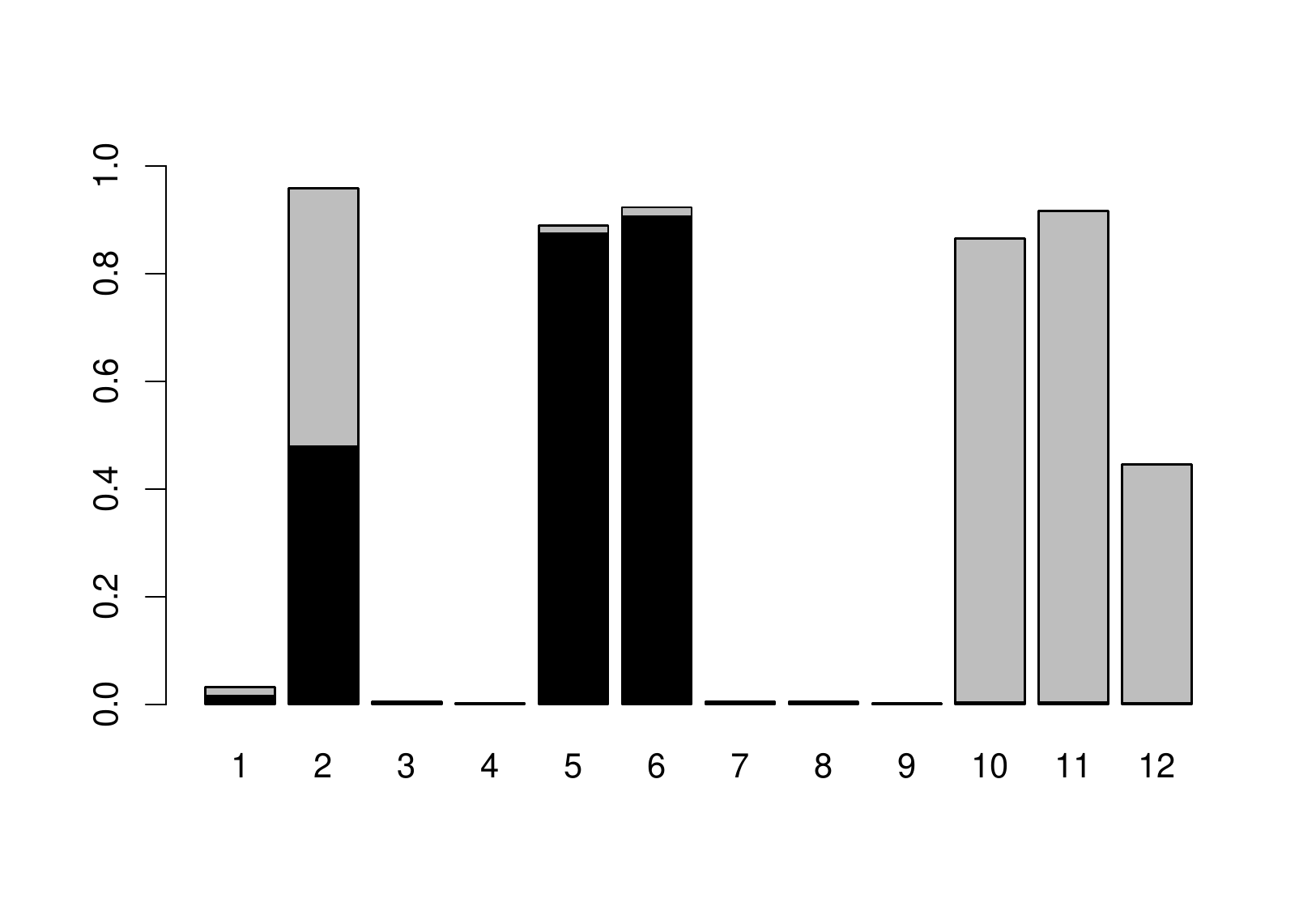}
	\end{center}
	\caption{\alex{Posterior probabilities for the carrier genotypes of each individual (Individual 1 to 12) in our hypothetical family (Figure~\ref{fig:fam4}). The posterior probability of being a paternal carrier $\P(X=10 | \mathrm{FH})$ (resp. maternal carrier $\P(X=01 | \mathrm{FH})$) is colored in black (resp. in grey). The deleterious allele being very rare in the general population ($f=0.33\%$), the probability of the monozygous carrier genotype is almost zero for each individual and it is therefore not represented here.}}\label{fig:PprobFam4}
\end{figure}

%We can see in Table~\ref{tab:post} the marginal posterior probability $\P(X_i=x | \mathrm{FH})$ for all individuals $i$ and all genotypes $x$ in our hypothetical family of Fig.~\ref{fig:fam4}.

\alex{Figure~\ref{fig:PprobFam4} represents the marginal posterior probability $\P(X_i=x | \mathrm{FH})$ for all individuals $i$ and for $x=10$ (paternal carrier) and $x=01$ (maternal carrier). Note that the posterior probability of the monozygous carrier genotype ($x=11$) being almost zero for each individual, it is not shown here. The posterior probability of the non-carrier genotype can be easily deduced.}

We can notice that the probabilities of being a non-carrier for 1, 3, 4, 7, 8 and 9 are all by far the highest despite the severe phenotype of relatives (granddaughter, niece or daughter). This result is consistent with the personal history of Individual 2 (ovarian cancer at age 51) which points her out as the most likely origin of the mutation in the family. Let us note that since we have no additional information on the ancestors of Individual 2, it is impossible to determine whether her mutation was transmitted by her father or her mother. As a consequence, the posterior carrier probability is equally shared between the paternal and maternal carrier genotypes.

Considering the severe personal history of cancer of Individuals 10 and 11, the most likely situation would be that they both received the mutation of their grandmother through their respective fathers (Individuals 6 and 5 respectively). The posterior probabilities are clearly consistent with this scenario: Individuals 5 and 6 have a probability of $\simeq 90\%$ to be maternal carriers, and Individuals 10 and 11 have similar probabilities to be paternal carriers. Note that Individual 12, being unaffected at age 37 (which is not very informative) basically have $50\%$ chance to have received the mutation from her father. 

%The deleterious allele being very rare in the general population ($f=0.0033$) and each individual having at least one parent with a low risk of being carrier, we unsurprisingly get only very low probabilities for the homozygous carrier genotype. %For this reason, the carrier probability is often assimilated to the heterozygous carrier probability which is false in general, unless under the lethal homozygous genotype assumption. 

\begin{figure}
\begin{center}
{\setlength{\tabcolsep}{20pt}
\begin{tabular}{ccc}
\includegraphics[width=0.2\textwidth,trim=0 0 0 50,clip]{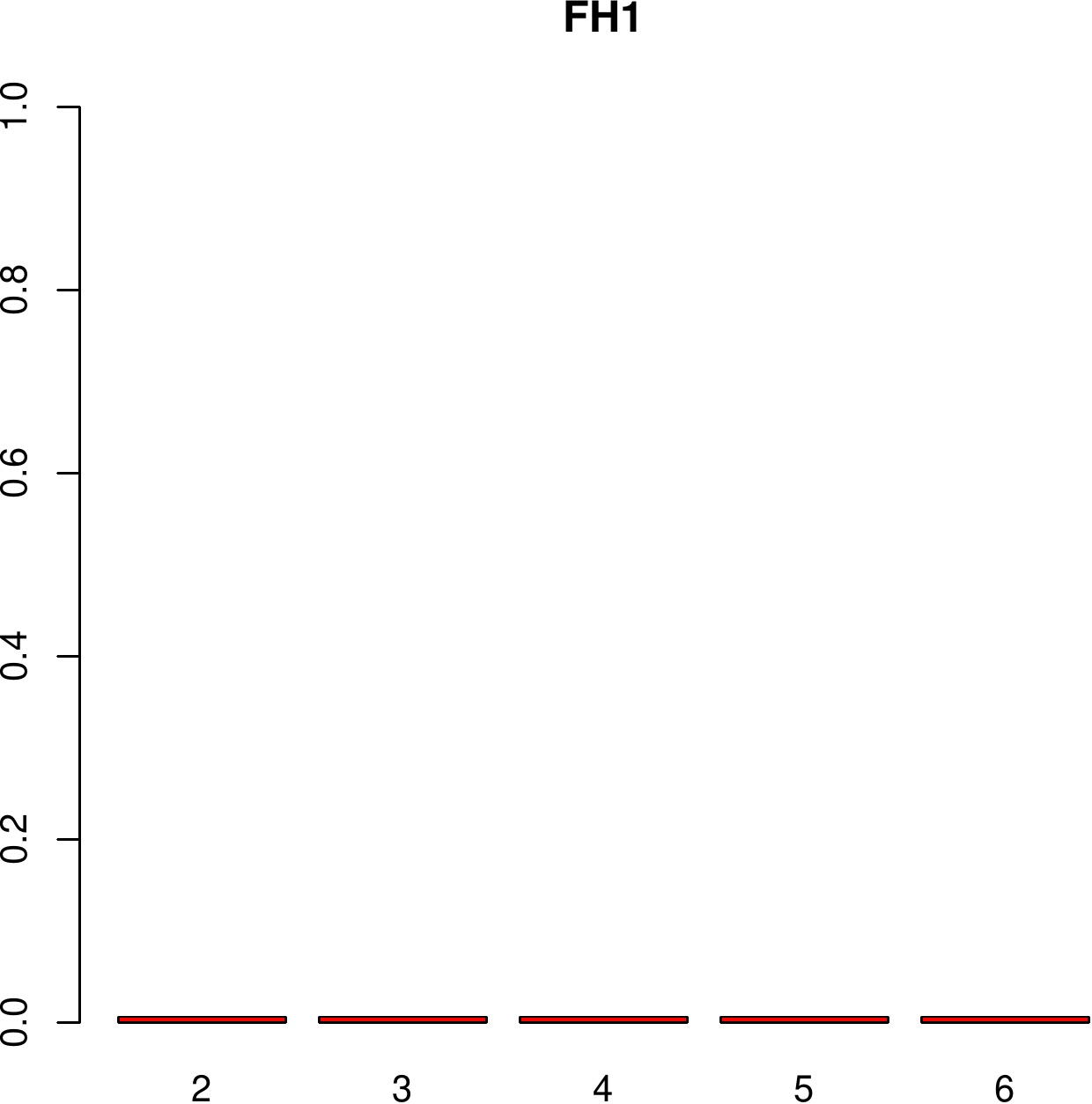}&
\includegraphics[width=0.2\textwidth,trim=0 0 0 50,clip]{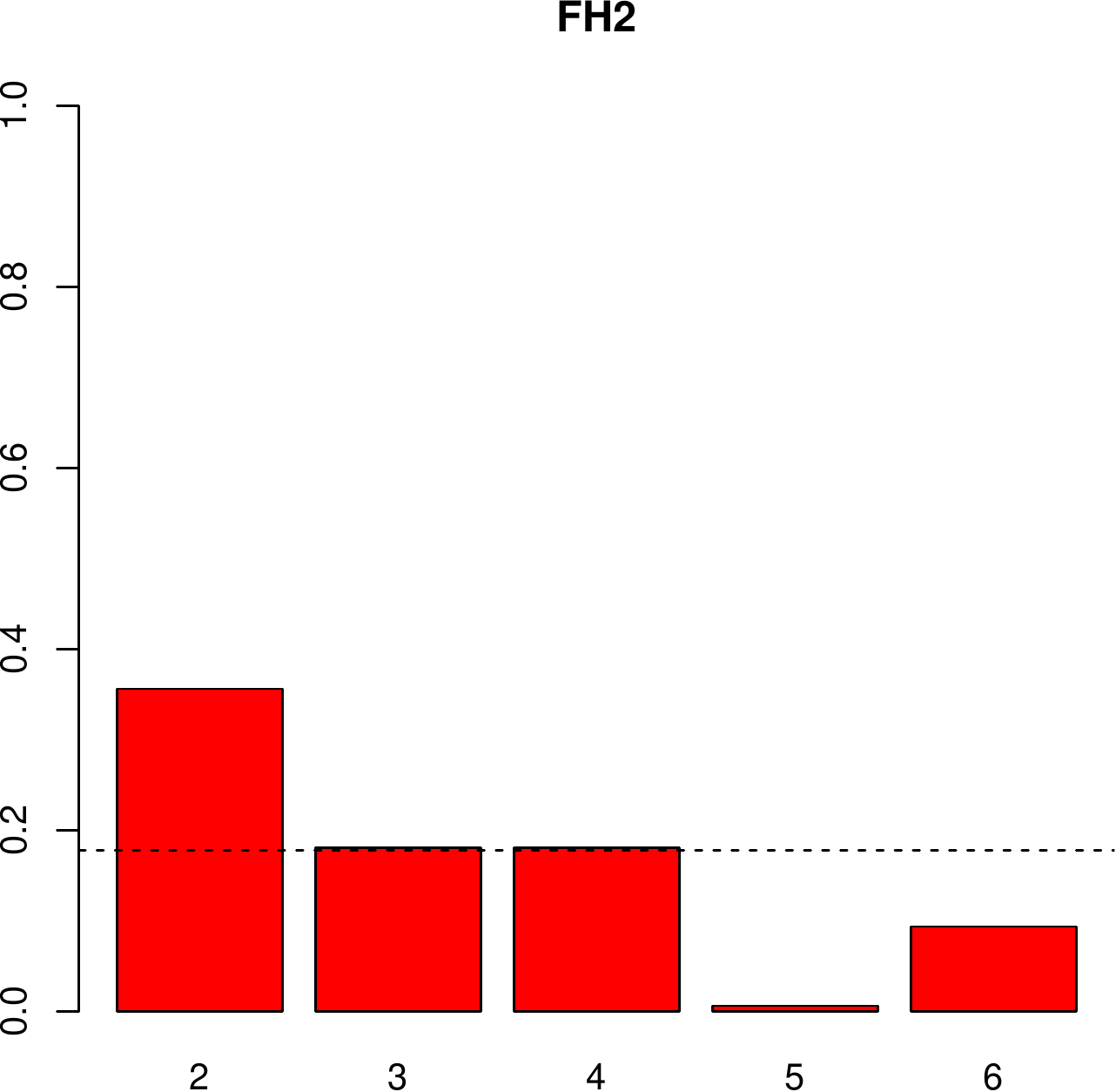}&
\includegraphics[width=0.2\textwidth,trim=0 0 0 50,clip]{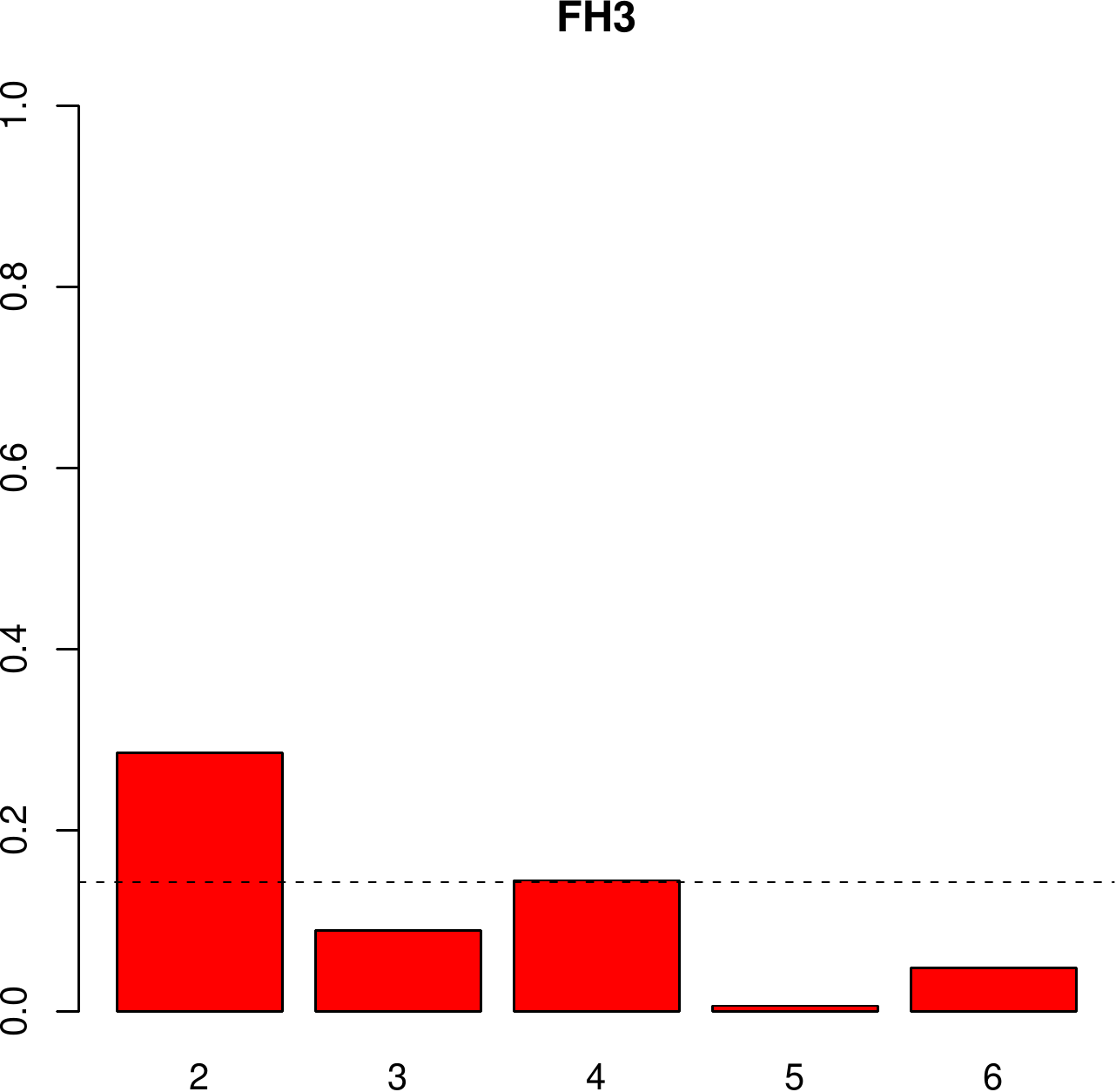}\\
\includegraphics[width=0.2\textwidth,trim=30 610 360 50,clip]{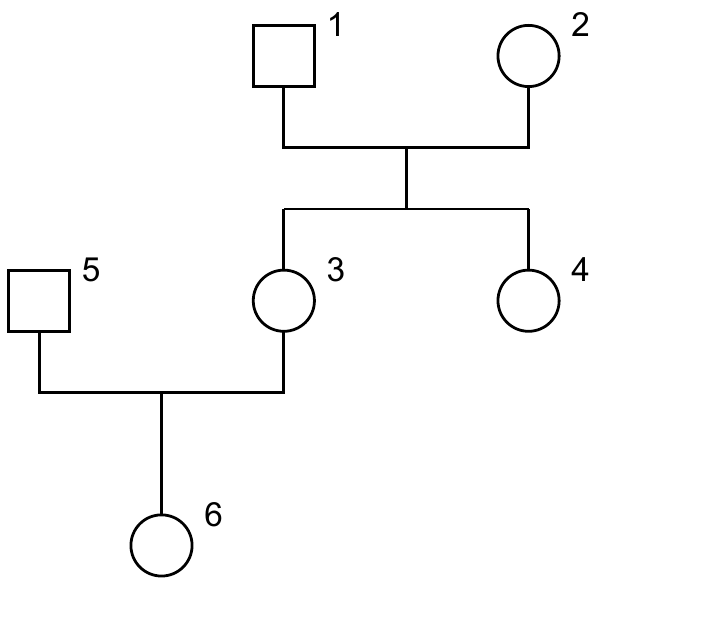}&
\includegraphics[width=0.2\textwidth,trim=30 610 360 50,clip]{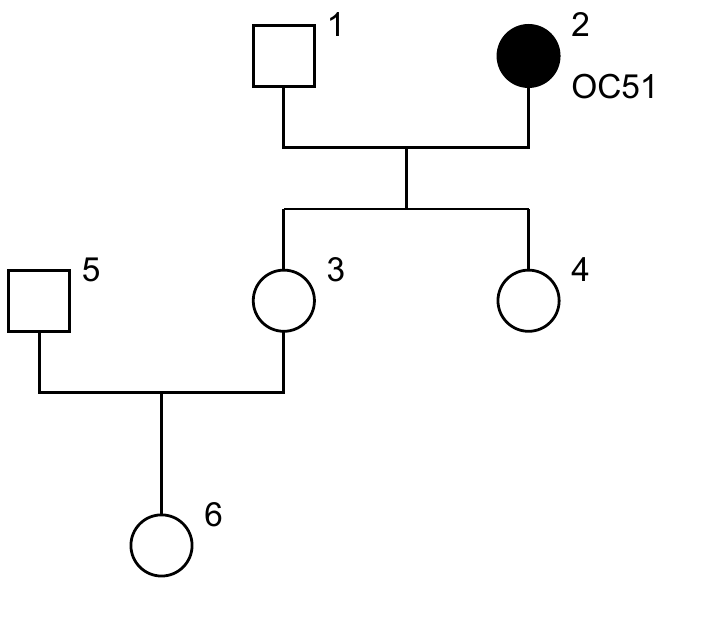}&
\includegraphics[width=0.2\textwidth,trim=30 610 360 50,clip]{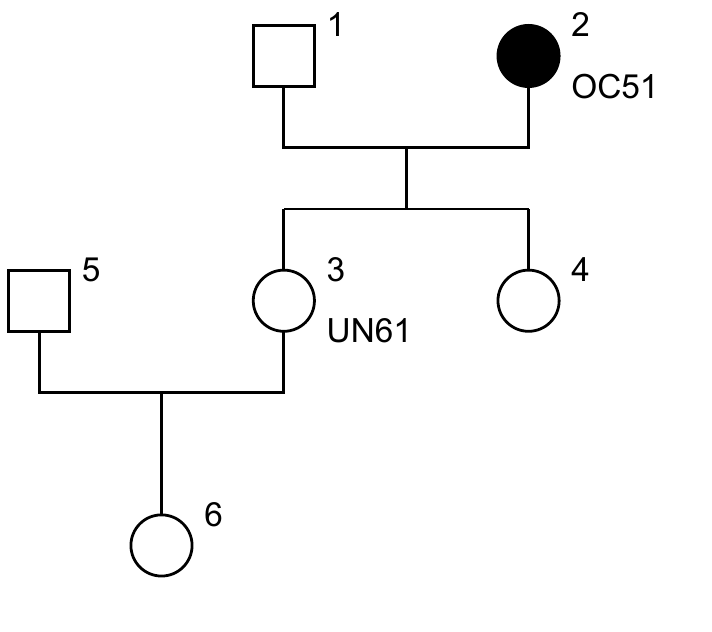}\\
FH1 & FH2 & FH3 \\
 \\
\includegraphics[width=0.2\textwidth,trim=0 0 0 50,clip]{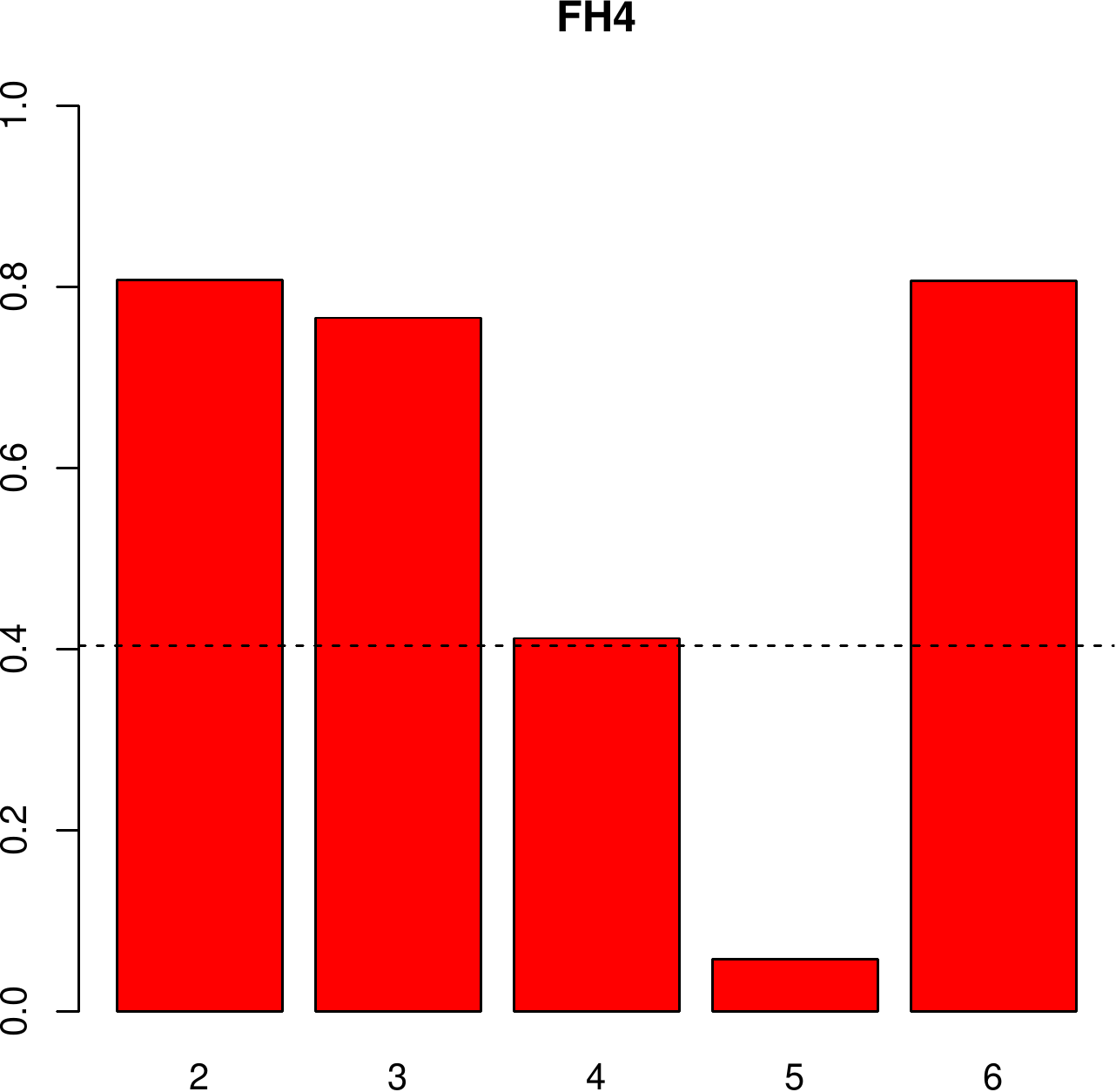}&
\includegraphics[width=0.2\textwidth,trim=0 0 0 50,clip]{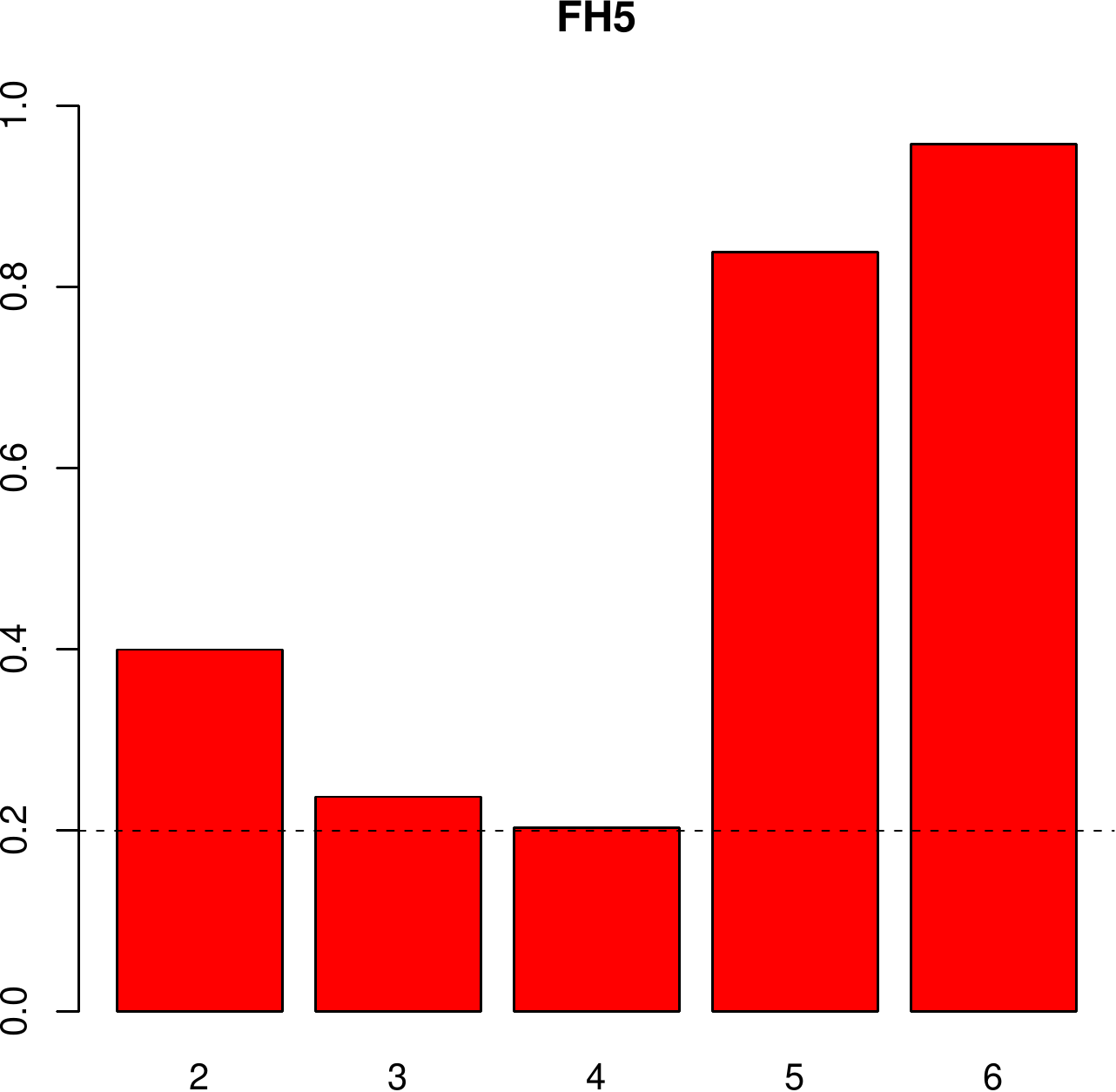}&
\includegraphics[width=0.2\textwidth,trim=0 0 0 50,clip]{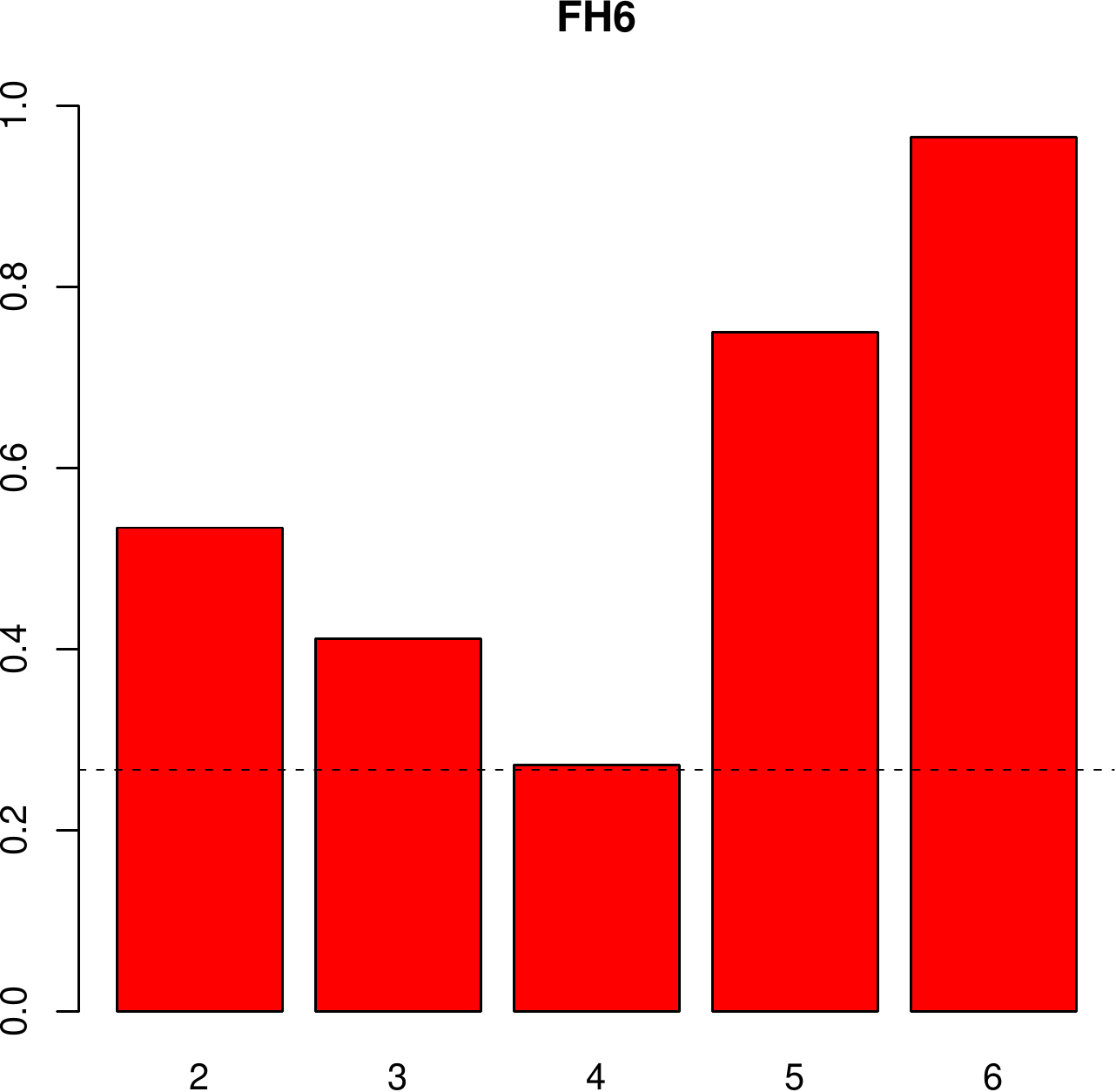}\\
\includegraphics[width=0.2\textwidth,trim=30 610 360 50,clip]{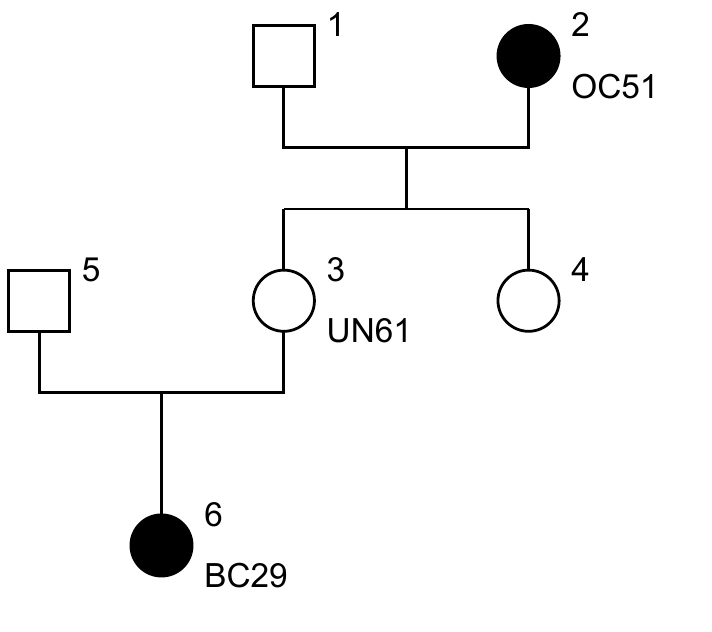}&
\includegraphics[width=0.2\textwidth,trim=30 610 360 50,clip]{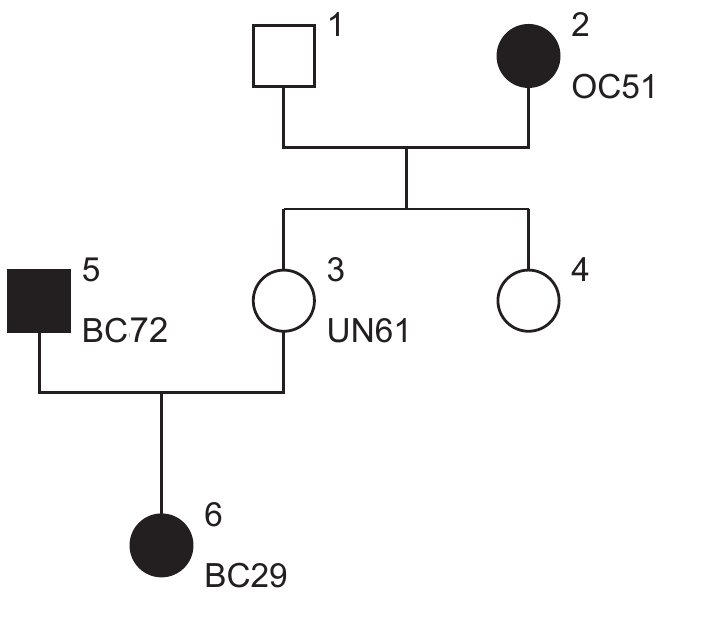}&
\includegraphics[width=0.2\textwidth,trim=30 610 360 50,clip]{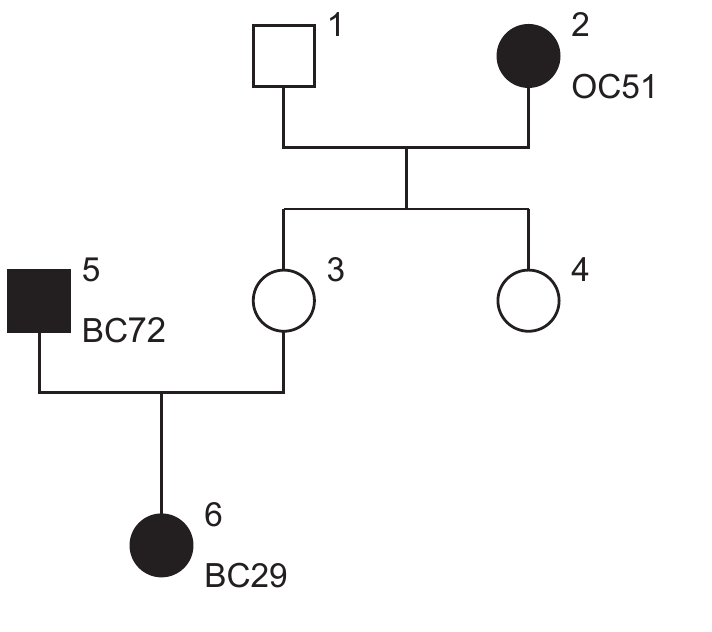}\\
FH4 & FH5 & FH6 \\
 %\\
%\includegraphics[width=0.3\textwidth,trim=30 610 360 50,clip]{FH6A}&
%\includegraphics[width=0.3\textwidth,trim=30 610 360 50,clip]{FH6B}&
%\includegraphics[width=0.3\textwidth,trim=30 610 360 50,clip]{FH6C}\\
%FH6a & FH6b & FH6c \\
\end{tabular}
}
\end{center}
\caption{Posterior marginal carrier distribution for a total of 6 FH with increasing degree of severity on the same pedigree structure with 6 individuals. Dashed line represent half the marginal carrier probability of Individual~2.}\label{fig:fh}
\end{figure}

Figure~\ref{fig:fh} shows some examples of the variations of the posterior marginal distribution of the genotypes in a same family structure according to different FH. We first notice that with no information (FH1) the posterior probabilities are exactly those of the general population: $\P(X_i \neq 00 | \mathrm{FH1})= 1-(1-f)^2 \simeq 0.0066$.

Note that Individual 2 has a severe personal history of cancer (ovarian cancer at age 51) in all other examples. As a consequence, Individual 1, as a male with no personal history of cancer, is mostly totally uninformative therefore not included in the forthcoming analyses.

Individual 4 having no children, she is independent from the rest of the family conditionally to her phenotype and her parent's genotype. With no information about her phenotype in any FH, her probability of being a carrier is therefore almost half her mother's one in each FH (because her father is almost uninformative). 
If we compare the posterior distribution of the genotype of Individual 3 in FH2, FH3 and FH4, we can notice that the ovarian cancer of her mother which increased her mother's probability of being a carrier raises her probability of being a carrier (FH2). A protective information about her phenotype such as no cancer until age 61 lowers her posterior probability of being a carrier (FH3). On the contrary, the cancer at young age of her daughter which increases her daughter's probability of being a carrier raises her own probability of being a carrier (FH4-6).

We also notice the causal relationships in a whole branch of the family with the transmission between Individuals 2, 3 and 6 of the deleterious allele being highly probable which raises the probability of being a carrier for Individual 3 even in the presence of a protective phenotype (unaffected at age 61) in FH4.

We finally observe the influence of the spouse's genotype when having children (FH5).
The higher risk of being a carrier for Individual 5 (because of his cancer at age 72) strongly decreases the carrier probability of his spouse (in comparison with FH4) since the paternal origin of the disease mutation naturally becomes the most likely event.
%The higher risk of being a carrier for Individual 5 (because of his cancer at 75) decreases his spouse one as the whole paternal branch for individual 6 is participating in the risk for Individual 6 of being a carrier. 
On the other side, the increase of risk for Individual 3 when suppressing her protective phenotype (FH6) also has a consequence on the marginal posterior distribution of her spouse in lowering his probability of being a carrier as his participation in the risk for their daughter is lowered. 

To summarize, one's probability of being a carrier mainly depends on: 1) one's probability of having at least one carrier parent, which is correlated to the history of cancer of one's ancestors; 2) one's probability of having transmitted the mutation to one's offspring which is correlated to the history of cancer of one's descendant relatives and one's spouse probability of being a carrier.

Remark: As introduced in the ``\nameref{sec:risk}'' section, we know that posterior carrier probabilities should decrease with time for unaffected individuals. For example, if we assume that Individual 4 is unaffected at age 40 in FH6, her probability of being a carrier is 24\%. If she stays unaffected up to age 60 (resp. age 80), her probability of being a carrier decreases to 15\% (resp. 8.5\%). 

\begin{table}
	\centering
	\begin{tabular}{lrrrrr}
		\hline
		$X_{2} / X_6$ & NC/NC & NC/C & C/NC & C/C\\
		\hline
		FH4\\
		marginal & 0.0371306 & 0.1551811 & 0.1559446 & 0.6517438\\
		joint  & 0.1443102 & 0.0480015 & 0.0487650 & 0.7589233\\
		\hline
		FH4 	and $X_3=10$\\
		marginal & 0.0092840 & 0.7741949 & 0.0025657 & 0.2139554\\
		joint & 0.0092840 & 0.7741949 & 0.0025657 & 0.2139554\\
		\hline
		FH4 and $X_3=01$\\
		marginal & 0.0000000 & 0.0000000 & 0.0118497 & 0.9881503\\
		joint & 0.0000000 & 0.0000000 & 0.0118497 & 0.9881503\\
		\hline
		FH4 and $X_3=11$\\
		marginal & 0.0000000 & 0.0000000 & 0.0000000 & 1.0000000\\
		joint & 0.0000000 & 0.0000000 & 0.0000000 & 1.0000000\\
		\hline
	\end{tabular}
	\caption{product of the posterior marginal probabilities $\mathbb{P}(X_2 | \mathrm{FH}) \mathbb{P}(X_6 | \mathrm{FH})$ and joint posterior probability $\mathbb{P}(X_2,X_6 | \mathrm{FH})$ in the context of known and unknown $X_3$. NC: non-carrier; C: carrier.}\label{tab:marginjoint}
\end{table}

Table~\ref{tab:marginjoint} gives a practical illustration of the dependence and conditional independence in a trio grandparent - parent - child. We compare the posterior joint distribution and the product of the posterior marginal distributions of genotypes $X_2$ and $X_6$ in FH4 with various information on $X_3$. We can see that these two quantities are not equal when $X_3$  is not observed while they are exactly the same when $X_3$ is fixed.  This example demonstrates how $X_2$ and $X_6$ are not conditionally independent given FH but they are, conditionally to FH and $X_3$. Note that when $X_3=11$, the mutation is necessarily found in both parents (Individual 1 and 2) as well as in her daughter (Individual 6).

\subsection{Cancer Risk}

As in Section~\ref{sec:risk} we now consider a female individual $i$ who is unaffected at age $\tau$ (\emph{i.e.} $\{{T}_i>\tau\} \subset \mathrm{FH}$) and denote by  $\pi=\P(X_i \neq 00 | \mathrm{FH})$ its posterior carrier probability. The purpose of this section is to compute the posterior risk of cancer for this individual (with or without the competing risk of death). As previously explained, these risks only depend on $\pi$ and $\tau$. 

\begin{figure}
	\begin{center}
		\includegraphics[width=0.8\textwidth,trim=0 0 0 0,clip]{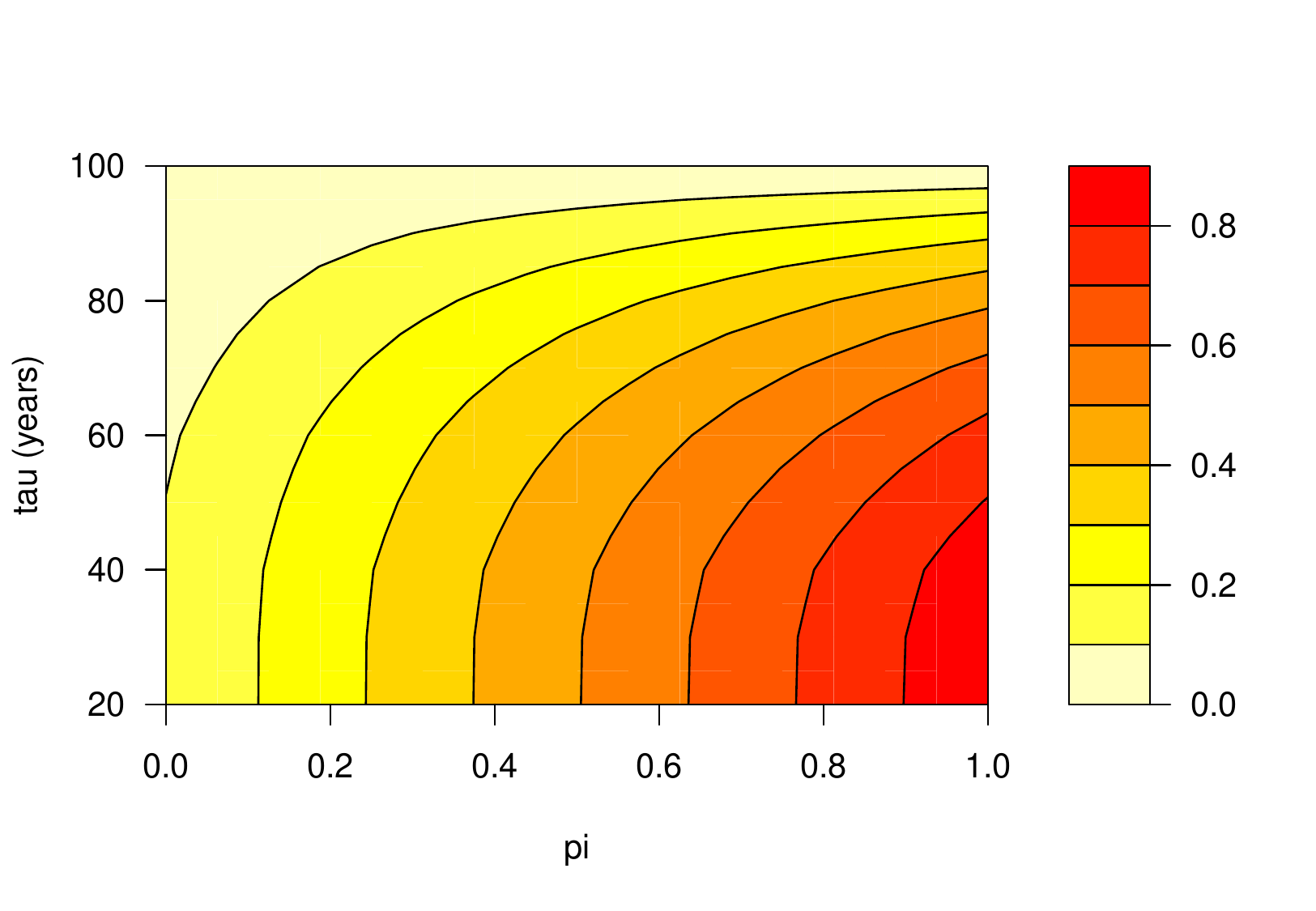}
	\end{center}
	\caption{Individual risk of breast cancer without the competing risk of death and for various $\pi$ and $\tau$}\label{fig:IndRiskPiTau}
\end{figure}

Figure~\ref{fig:IndRiskPiTau} represents the individual risk of breast cancer up to age 100\footnote{Note that we obtain qualitatively similar results with a lower age limit (\emph{e.g.} age 80), but quantitative results are more illustrative with age 100.} without the competing risk of death and variant $\pi$ and $\tau$. We can see that the individual risk of BC rises as $\pi$ increases and $\tau$ decreases. This result is quite intuitive as the  younger a patient is, the longer she will be at risk until age 100; the greater her probability of carrying a deleterious allele, the greater her risk to develop a cancer.

%%%%% where ???
As introduced in the previous section the probability of being a carrier for an unaffected individual decreases with time if she stays unaffected. Assuming Individual 4 was 52 in FH4, Figure~\ref{fig:PcarrierWtime} shows the evolution of the probability of being a carrier for Individual 3 and Individual 4 in FH4. As they stay unaffected we can clearly see the decrease of this probability which has to be taken into account in the computation of the individual risk of breast cancer over time (see Section~\ref{sec:risk}).

\begin{figure}
	\begin{center}
		\includegraphics[width=0.8\textwidth,trim=0 0 0 0,clip]{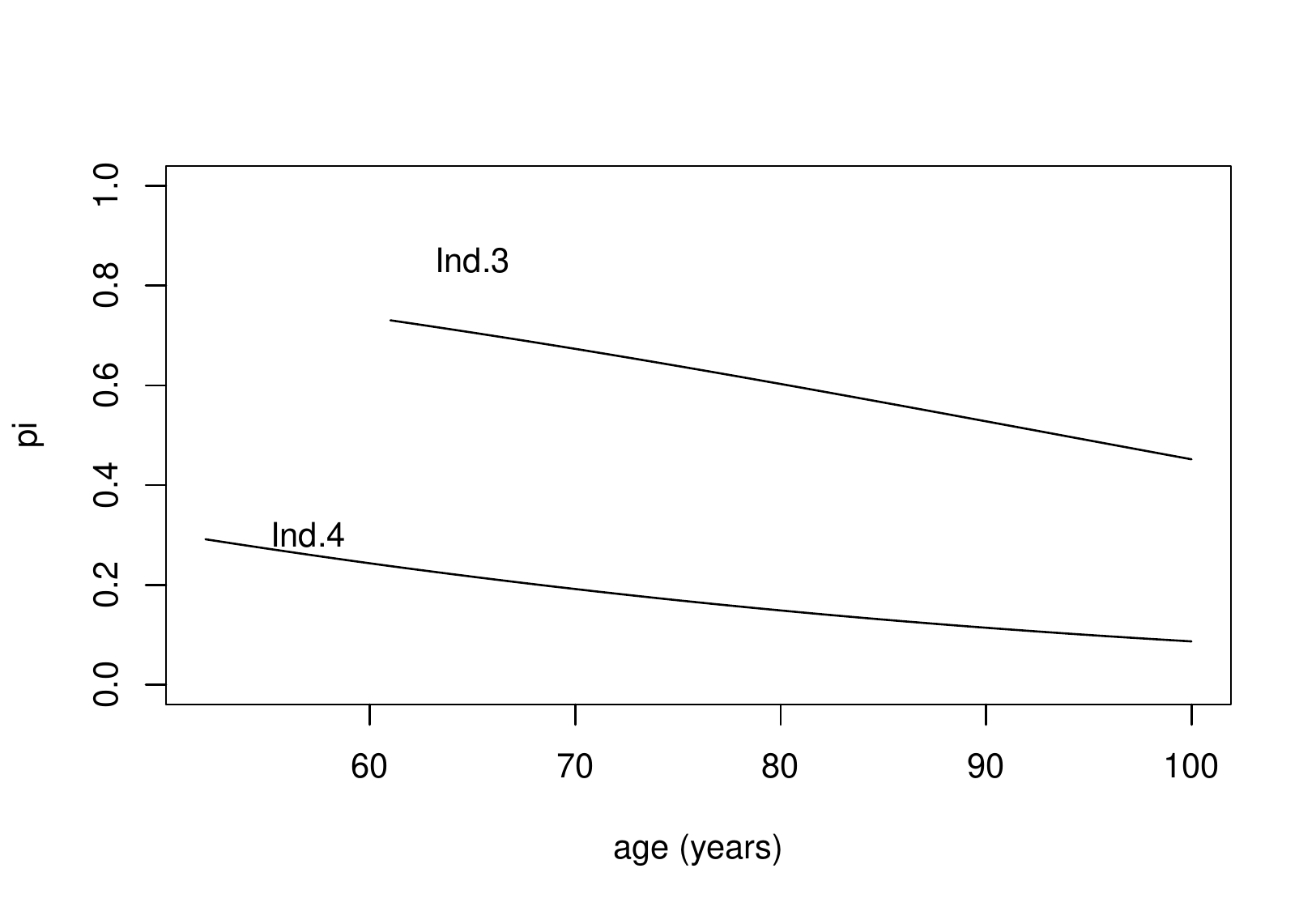}
	\end{center}
	\caption{Posterior probabilities of being a carrier according to the time for Individuals 3 and 4 in FH4 assuming Individual 4 is 52 at the time of the censoring.}\label{fig:PcarrierWtime}
\end{figure}

%However, as explained before, in order have meaningful estimate the individual risk of BC it is critical to consider the competing risk of death. 

As explained in Section~\ref{sec:risk}, computing risk with the competing risk of death requires a numerical discretization of age by a fixed step $\Delta t$.  In order to calibrate $\Delta t$ we used $\Delta t=0.01$ as a reference, and observed that $\Delta t=0.1$ is a reasonable balance between accuracy and computational efficiency (data not shown).

%To compute the individual instant risk of breast cancer with competing risk with death, we need to discretize the time for each individual from $\tau$ to a chosen $t_{max}$. (See section~\ref{sec:...}). We take a hundredth of a year as a standard and we compute the survival with different steps of discretization from a year to a hundredth of a year. We calculate the survival from 40 to 100 years each year for an unaffected woman who's $\tau=40 \, \text{years}$ and $\pi=0.5$ with the different computed survival functions. The mean square error between the results obtained with a discretization of a tenth of a year and the standard discretization is $1.8e-4$. We will use a discretization of a tenth of a year.  

\begin{figure}
	\begin{center}
		\includegraphics[width=0.8\textwidth,trim=0 0 0 0,clip]{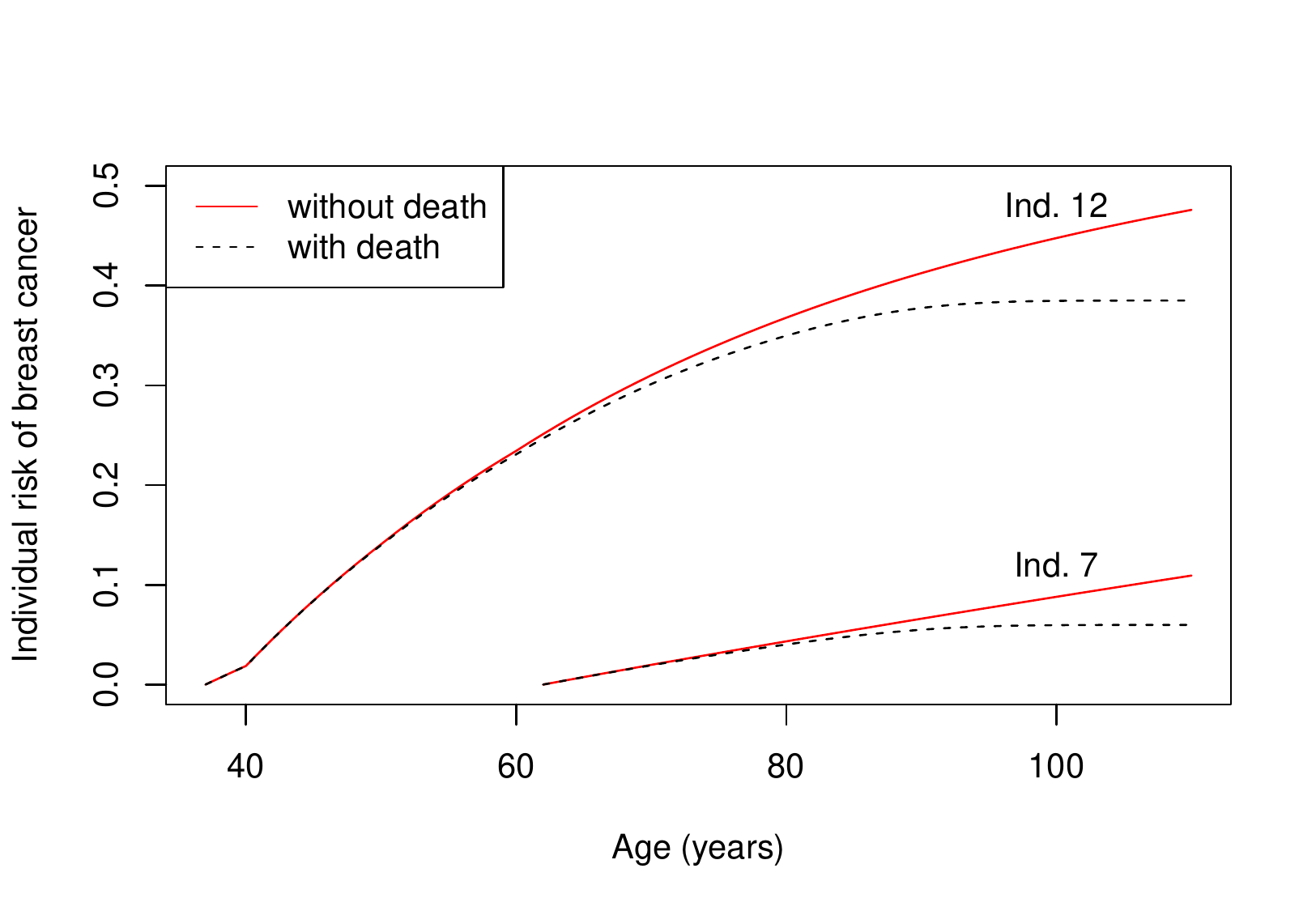}
	\end{center}
	\caption{Individual risk of breast cancer with and without the competing risk of death for individual 7 and 12 of our hypothetical family from $\tau$ to 100 years with and without the competing risk of death.}\label{fig:Compdeathhypofam}
\end{figure}

Figure~\ref{fig:Compdeathhypofam} represents the individual risk of breast cancer for Individual  7 ($\pi=0.553 \%$ and $\tau=62$ years) and Individual 12  ($\pi=44.6 \%$ and $\tau=37$ years) in our hypothetical family from $\tau$ to 100 years with and without taking into account the competing risk of death. We can see that the difference between the two curves for each individual is increasing with the age. The age from which the difference becomes significant varies with the couple ($\pi, \tau$). We also observe that the individual risk of breast cancer eventually reaches a plateau which corresponds to the point where the incidence of breast cancer becomes negligible compared to the incidence of death in the elderly.

\begin{figure}
	\begin{center}
		\includegraphics[width=0.8\textwidth,trim=0 0 0 0,clip]{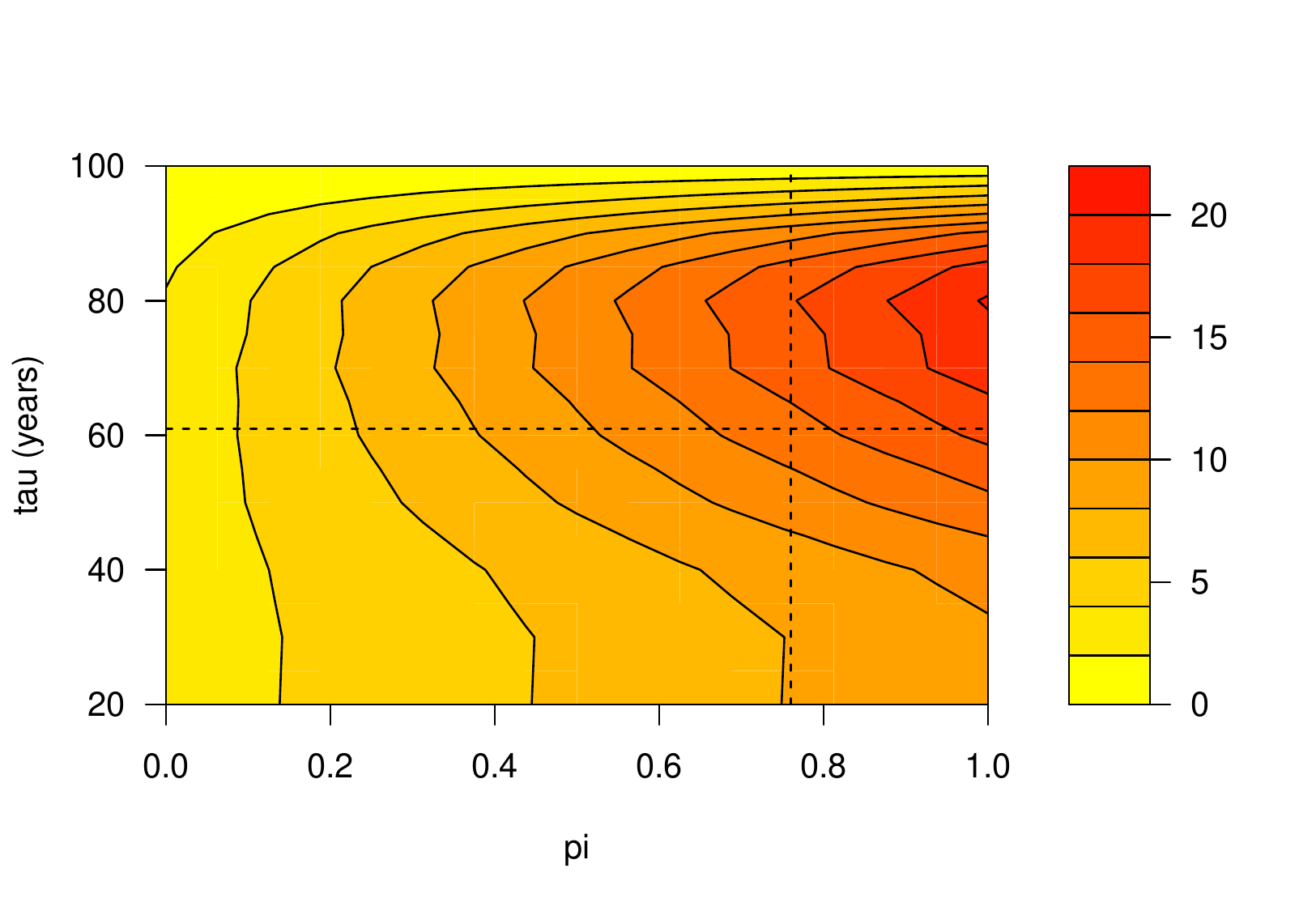}
	\end{center}
	\caption{Difference (in percentage) between the individual risk of breast cancer up to 100 years without and with the competing risk of death for various $\pi$ and $\tau$. Specific values $\pi=76.59 \%$ and $\tau = 61$ are given by the dashed lines. }\label{fig:heatmapDiff}
\end{figure}

Quantitatively, the importance of taking into account the competing risk of death is pointed out in the Figure~\ref{fig:heatmapDiff} which represents the difference between the individual risk of breast cancer up to the age of 100 years for variant couples ($\pi$, $\tau$). For example for Individual 3 in FH4 ($\pi=76.59 \%$, $\tau = 61$, see Figure~\ref{fig:fh}), the error while calculating her individual risk of breast cancer up to the age of 100 years reaches almost 14 \%. If it is clear that the competing risk of death can have a limited effect on the global risk of cancer for certain couples $(\pi,\tau)$ its effect is never totally negligible, and since we provide a rigorous way to take it into account we strongly advocate its use in all circumstances.

\section{Conclusions}

We presented here a general model for genetic disease with variable age at onset. This model, a Bayesian network, combines classical genetic modeling with survival analysis. In order to deal with the (mostly) unobserved genotypes, we first explained in detail how belief propagation can be used to perform likelihood and posterior probability computations. Secondly, we focused on the challenging problem of computing posterior individual disease risks, with or without taking into account the competing risk of death. Finally, we illustrated theses results with the Claus-Easton model for breast and ovarian cancer.
%\rev{Note that the (unoptimised and basically uncommented) R source code for the computations/simulations of the manuscript is available on request for the (brave) interested reader.}
\rev{The R source codes are available upon request for the interested readers.}
%\olivier{Peut etre un peu trop familier pour un papier ? A mon avis on met soit : "The R source codes are available upon request for the interested readers". Ou alors, on met rien, ou alors on met: "The R source codes are available on the following webpage $\ldots$" et on pointe sur la page de l'un d'entre nous. Pas sûr qu'on ait besoin que ce soit ultra propre et commenté, il y a un exemple sur ma page web pour mon papier à biostatistics (où les codes sont très mal commentés), ça donnera au moins l'impression qu'on a fait un effort.}

For the sake of simplicity, we only considered a bi-allelic locus with standard distribution (autosomal, Hardy-Weinberg, Mendelian allele transmission) but extensions (\emph{e.g.} multi-loci, unbalanced allele transmission, lethal genotypes, etc.) are straightforward. For the survival model, we presented a simple dominant effect without covariates, but again, extensions to any proportional hazard model (\emph{e.g.} recessive, additive, with covariates, etc.) are easy to implement. Incorporating random effects (at the individual and/or familial level) in the model \citep[like in the BOADICEA model, see][]{antoniou2002comprehensive,antoniou2004boadicea} is clearly also possible, but slightly more challenging.

Computation of posterior carrier distributions remains almost unchanged except for the random effect support which must be discretized (five values are claimed to be sufficient in the BOADICEA literature) and for the belief propagation which must be performed once for each of the possible value of the random effect. For posterior risks, calculations get slightly more complex since the posterior individual hazard must now be integrated over the (changing over time) posterior joint distribution of the individual genotype and of the random effect. Basically, all computations are slightly more intensive with random effects, but most results of Section~\ref{sec:risk} remain very similar.

One of the important limitations of the present work is the fact that we assume that all model parameters are known. However, it should be noted that likelihood and conditional likelihood might be easy to compute through the belief propagation which means that we basically provide all the necessary means to estimate the model parameters from actual data. In that context, it is nevertheless critical to deal efficiently with ascertainment issues: the fact that the family ending up in the database are usually precisely the one with the most severe disease family history. But standard methods like the PEL \citep{alarcon2009pel}, which basically are conditional likelihood computations, are known to deal relatively well with the problem.

In order to take into account the competing risk of death, we used death from all causes, which was obtained from registry data \citep{ined}. However, only death without cancer precludes the onset of cancer and we are not interested into death from all causes. Since registry data usually do not report the causes of death it is a difficult task to estimate the risk of death without cancer. This has been studied for instance in \cite{wanneveich2016} through a illness-death model, using registry data and differential equations to model the specific causes of death. Nevertheless, it is very likely that the gain in terms of predictions would be minor as mortality from all causes is likely to be close to mortality without cancer. 

Further work includes all the extensions described above (\emph{e.g.} more complex genetic model, genetic tests, familial random effects, etc.) as well as the development of a clinical web application for the Claus-Easton model in close collaboration with the cancer genetics department of the \emph{Institut Curie}. From the methodological point of view, we plan to focus on the computation of more complex posterior distribution like the number of carriers in any subgroup of individuals and/or the familial posterior risk (time before any family member at risk is diagnosed).

\section*{Acknowledgments}

We would first like to thank both anonymous reviewers for their constructive comments and remarks.
This work received the support of both the \emph{Institut de Recherche en Santé Publique} (IRESP) and the \emph{Ligue National Contre le Cancer} (LNCC). Alexandra Lefebvre's internship was funded by the \emph{Institut Curie}. We finally warmly thank Antoine de Pauw (\emph{Institut Curie}) for his continuous friendly support and for suggesting the hypothetical family presented here.

\section*{Conflict of Interest}

The author(s) declare(s) that there is no conflict of interest regarding the publication of this paper.

\appendix

%\section{Building Junction Tree}\label{appendix:jt}

%\greg{Reference for NP-hard. Then Mininum fill-in heuristic. Related to the notion of \emph{cutsets} in genetics.}

\section{Proofs for the Carrier Risk Section}\label{appendix:bp}

For all $k \in \{1,\ldots,K\}$ we recursively define: $u_k= \{k\} \cup_{j \in \mathrm{from}_k} u_j$, 
$U_k = \cup_{j \in u_k} C_j$, and $V_k = \cup_{j \notin u_k} C_j$. Then we can compute the so-called \emph{forward} and \emph{backward} quantities over any \rev{separator $S_j=C_j \cap C_{\mathrm{to}_j}$}:
$$
F_j(S_j)=\sum_{  U_j \setminus S_j } \prod_{X_i \in U_j^*} K_i \left(X_i | X_{\mathrm{pa}_i} \right) 
\quad\text{and}\quad
B_j(S_j) = \sum_{V_j \setminus S_j} \prod_{X_i \in V_j^*} K_i \left(X_i | X_{\mathrm{pa}_i} \right)
$$
where $U_j^* = \{X_i \in U_j, \exists k \in u_j, \mathrm{of}_i=k \}$  and $V_j^* = \{X_i \in V_j, \exists k \notin u_j, \mathrm{of}_i=k \}$.

The key is then to prove that, for all $j\in \{1,\ldots,K\}$ we have:
\begin{equation}\label{eq:fb1}
\P(S_j , \text{FH})= F_j(S_j)B_j(S_j)
\end{equation}
\begin{equation}\label{eq:fb2}
\P(C_k , \text{FH})=  \Phi_k(C_k)  \times \prod_{j \in \mathrm{from}_k}  F_j(S_j)  \times B_k(S_k).
\end{equation}

For proving Eq.~(\ref{eq:fb1}), we start by noticing that the \rev{JT (Junction Tree)} properties \citep{koller2009probabilistic} give:
$\{X_1,\ldots,X_n\} \setminus S_j = (U_j \setminus S_j)\uplus (V_j\setminus S_j)$ and $\{X_1,\ldots,X_n\} = U_j^* \uplus V_j^*$ (both being disjoint unions). We therefore have:
\begin{eqnarray*}
\P(S_j , \text{FH})&=& \sum_{  U_j \setminus S_j } \sum_{V_j \setminus S_j}
\prod_{X_i \in U_j^*} K_i \left(X_i | X_{\mathrm{pa}_i} \right)
\prod_{X_i \in V_j^*} K_i \left(X_i | X_{\mathrm{pa}_i} \right)\\
&=& 
\underbrace{ \left( \sum_{  U_j \setminus S_j } \prod_{X_i \in U_j^*} K_i \left(X_i | X_{\mathrm{pa}_i} \right) \right) }_{F_j(S_j)} \times 
\underbrace{ \left(\sum_{V_j \setminus S_j} \prod_{X_i \in V_j^*} K_i \left(X_i | X_{\mathrm{pa}_i} \right) \right)}_{B_j(S_j)}
\end{eqnarray*}
the factorization between the first and second equation being possible thanks to the fact that
$\left(\cup_{X_i \in U_j^*} \{ X_i , X_{\mathrm{pa}_i} \}\right)
\cap
\left(\cup_{X_i \in V_j^*} \{ X_i , X_{\mathrm{pa}_i} \}\right) = S_j$ (JT properties again).

%\alex{J'espère ne pas me tromper mais il me semble qu'il faut rajouter qu'on peut factoriser car:	$U_j \cap V_j = S_j$ and therefore $(U_j \setminus S_j)\cap(V_j \setminus S_j) = \emptyset$	 (running intersection) and $U_j^* \subseteq U_j$ and $V_j^* \subseteq V_j$.}

The proof is basically the same for Eq.~(\ref{eq:fb2}) using $\{X_1,\ldots,X_n\} \setminus C_k = \uplus_{j \in \mathrm{from}_k} (U_j \setminus S_j) \uplus (V_k \setminus S_k)$ \rev{we get:}
%\begin{eqnarray*}
%\mathbb{P}(C_k,\mathrm{FH})&=&  \Phi_k(C_k) \sum_{j \in \mathrm{from}_k} \sum_{  U_j \setminus S_j } \sum_{  V_k \setminus S_k }  \prod_{j \in \mathrm{from}_k} \prod_{X_i \in U_j^*} K_i \left(X_i | X_{\mathrm{pa}_i} \right)
%\prod_{X_i \in V_k^*} K_i \left(X_i | X_{\mathrm{pa}_i} \right)\\
%&=& \Phi_k(C_k)  \prod_{j \in \mathrm{from}_k} \underbrace{\sum_{  U_j \setminus S_j }  \prod_{X_i \in U_j^*} K_i %\left(X_i | X_{\mathrm{pa}_i} \right)}_{F_j(S_j)} 
%\underbrace{\sum_{  V_k \setminus S_k }  \prod_{j \in \mathrm{from}_k} \prod_{X_i \in U_j^*} K_i \left(X_i | X_{\mathrm{pa}_i} \right)}_{B_k(S_k)}.
%\end{eqnarray*}
%}
%\alex{Il y a une somme qui m'échappe car j'aurais plutôt vu ça comme un produit (avec un union sur j de blocs disjoints qui se factorisent) mais je ne suis pas très sûre et pas trop sûre de comment expliquer ce que je veux dire et par contre, pour le coup, plus sûre : il y a un swap entre $U_j^*$ et $V_k^*$ dans le dernier backward mais d'une manière générale, ça ma parait plus simple comme ça, qu'est-ce que vous en pensez (si c'est juste) :}
\alex{
\begin{eqnarray*}
	\mathbb{P}(C_k,\mathrm{FH}) &=& \sum_{\{X_1,\ldots,X_n\} \setminus C_k} \prod_{X_i \in \{X_1,\ldots,X_n\}} K_i \left(X_i | X_{\mathrm{pa}_i} \right)\\
	&=&  \Phi_k(C_k) \prod_{j \in \mathrm{from}_k} \sum_{  U_j \setminus S_j } \sum_{  V_k \setminus S_k } \prod_{X_i \in U_j^*} K_i \left(X_i | X_{\mathrm{pa}_i} \right)
	\prod_{X_i \in V_k^*} K_i \left(X_i | X_{\mathrm{pa}_i} \right)\\
	&=& \Phi_k(C_k)  \prod_{j \in \mathrm{from}_k} \underbrace{\sum_{  U_j \setminus S_j }  \prod_{X_i \in U_j^*} K_i \left(X_i | X_{\mathrm{pa}_i} \right)}_{F_j(S_j)} 
	\underbrace{\sum_{  V_k \setminus S_k } \prod_{X_i \in V_k^*} K_i \left(X_i | X_{\mathrm{pa}_i} \right)}_{B_k(S_k)}.
\end{eqnarray*}
}
\alex{The factorisation being possible as 
 $\uplus_{j \in \mathrm{from}_k} (U_j \setminus S_j) 
 \cap 
 (V_k \setminus S_k) = \emptyset$ (running intersection) 
and 
$\forall j, \forall k, U_j^* \subseteq U_j$
and
$V_k^* \subseteq V_k$.
}

Finally, the recursive expression of the forward and backward quantities can be easily derived from equations (\ref{eq:fb1}) and (\ref{eq:fb2}):
\rev{
\begin{eqnarray*}
\P(S_k , \text{FH}) &= & \sum_{C_k \setminus S_k}  \P(C_k , \text{FH}) \\
F_k(S_k)\cancel{B_k(S_k)} 
& =&  \sum_{C_k \setminus S_k} \prod_{j \in \mathrm{from}_k}  F_j(S_j) \times \Phi_k(C_k)  \times \cancel{B_k(S_k)}
\end{eqnarray*}
}
which gives the forward recursion by simplifying the $B_k(S_k)$ term.

\section{Proofs for the Disease Risk Section} \label{appendix:risk}

\begin{proof}[Proof of Theorem \ref{th:survnocr}]

\olivier{For clarity, we recall that $S_0(t)=\mathbb{P}(T_i > t |  X_i=00 )$, $S_1(t)=\mathbb{P}(T_i > t |  X_i\neq00 )$, $\pi=\P(X_i \neq 00 | \mathrm{FH},{T}_i>\tau)$ and $S(t)=\P({T}_i>t|\mathrm{FH},{T}_i>\tau)$, for $i=1,\ldots,n$, and that $\{T_i>\tau\} \subset \mathrm{FH}$. }\greg{Since the $T_i$ are independent conditionally to the $X_i$, the distribution of $T_i$ conditionally on $X_i$ obviously does not depend on $\mathrm{FH}$ (for values of $X_i$ which are not forbidden by $\mathrm{FH}$). This is why $\mathrm{FH}$ can be omitted almost everywhere  in the following proof as soon as $\pi$ has been computed.}
%As a consequence, $\mathrm{FH}$ can be omitted in the conditioning: $\mathbb{P}(T_i>t|X_i,\mathrm{FH})=\mathbb{P}(T_i>t|X_i)$
%In the following proofs we repeatedly use the argument that\alex{, once $\pi$ is computed,} the family history does not provide any further information on the genotype $X_i$ to predict an individual risk.}  

We have $S(t)=\sum_{X_i} \P(T_i>t,{X_i}|T_i>\tau,\mathrm{FH})$, \olivier{where the notation $\sum_{X_i}$ represents the summation over the different possible values of $X_i$, that is $X_i=00$ or $X_i\neq 00$}. Using Bayes' rule, 
\begin{align*}
\P(T_i>t,X_i\neq 00|T_i>\tau,\mathrm{FH})&=\P(T_i>t|X_i\neq 00,T_i>\tau,\mathrm{FH})\times\P(X_i\neq 00|T_i>\tau,\mathrm{FH})\\
& = \frac{\P(T_i>t,X_i\neq 00,\mathrm{FH})}{\P(T_i>\tau,X_i\neq 00,\mathrm{FH})}\times \pi \\
& = \frac{\P(T_i>t|X_i\neq 00,\mathrm{FH})}{\P(T_i>\tau|X_i\neq 00,\mathrm{FH})}\times \pi=\frac{S_1(t)}{S_1(\tau)}\pi,
\end{align*}
\olivier{where we used the fact that $\P(T_i>t|X_i\neq 00,\mathrm{FH} )=\P(T_i>t|X_i\neq 00 )$}.
We similarly prove that $\P(T_i>t,X_i= 00|T_i>\tau,\mathrm{FH})=(1-\pi)S_0(t)/S_0(\tau).$

The next result is proved using Bayes' rule:
\begin{align*}
\P(X_i \neq 00 | \mathrm{FH}, T_i>t)&=\frac{\P(X_i \neq 00 , \mathrm{FH}, T_i>t)}{\P(\mathrm{FH}, T_i>t)}\\
&=\frac{\P(T_i>t|X_i \neq 00,T_i>\tau )}{\P(T_i>t|\mathrm{FH},T_i>\tau)}\P(X_i \neq 00 |\mathrm{FH},T_i>\tau),
\end{align*}
\olivier{where we also used the fact that $\P(T_i>t|X_i \neq 00,\mathrm{FH},T_i>\tau )=\P(T_i>t|X_i \neq 00,T_i>\tau )$}. 

We then directly have $\P(T_i>t|X_i \neq 00,T_i>\tau )=S_1(t)/S_1(\tau)$ from Bayes' rule, $\P(X_i \neq 00 |\mathrm{FH},T_i>\tau)=\pi$ and $\P(T_i>t|\mathrm{FH},T_i>\tau)=S(t)$ which concludes the proof.

Finally, \olivier{in order to prove Equation~\eqref{eq:posthaz}, we recall that }%, is proved writing:
\begin{align*}
\lambda(t)&=\lim_{\Delta t\to 0}\frac{\P(t\leq T_i<t+\Delta t | T_i\geq t,\mathrm{FH})}{\Delta t}\\
\lambda_0(t)&=\lim_{\Delta t\to 0}\frac{\P(t\leq T_i<t+\Delta t | T_i\geq t,X_i= 00)}{\Delta t}\\
\lambda_1(t)&=\lim_{\Delta t\to 0}\frac{\P(t\leq T_i<t+\Delta t | T_i\geq t,X_i\neq 00)}{\Delta t}
\end{align*}
Then, 
\begin{align*}
\P(t\leq T_i<t+\Delta t | T_i\geq t,\mathrm{FH})&=\sum_{X_i} \P(t\leq T_i<t+\Delta t,X_i | T_i\geq t,\mathrm{FH})\\
&=\sum_{X_i} \P(t\leq T_i<t+\Delta t,X_i,\mathrm{FH})/\P(T_i\geq t,\mathrm{FH})\\
&= \sum_{X_i} \P(t\leq T_i<t+\Delta t|X_i ) \P(X_i | T_i\geq t,\mathrm{FH}),
\end{align*}
\olivier{using Bayes' rule and the fact that $\P(t\leq T_i<t+\Delta t|X_i, \mathrm{FH})=\P(t\leq T_i<t+\Delta t|X_i )$ and $\P(X_i,\mathrm{FH} | T_i\geq t,\mathrm{FH})=\P(X_i | T_i\geq t,\mathrm{FH})$}. Dividing by $\Delta t$ and taking the limit as $\Delta t$ tends to $0$ gives
\begin{align*}
\lambda(t)=\lambda_1(t)\times \P(X_i\neq 00 | T_i\geq t,\mathrm{FH})+\lambda_0(t)\times \P(X_i= 00 | T_i\geq t,\mathrm{FH})
\end{align*}
We showed previously that $\P(X_i\neq 00 | T_i\geq t,\mathrm{FH})=\pi S_1(t)/(S(t)S_1(\tau))$ and $\P(X_i= 00 | T_i\geq t,\mathrm{FH})=(1-\pi) S_0(t)/(S(t)S_0(\tau))$ which concludes the proof.
\end{proof}

\begin{proof}[Proof of Lemma \ref{th:survcr}]
The first part of the equality is a standard result in the competing risk setting: we have, from Bayes' rule,
\begin{align*}
\lambda_\alpha(u) & =\lim_{\Delta t\to 0}\frac{\P(t\leq T_i<t+\Delta t |\mathrm{FH})}{\Delta t\,\P(T_i^*\geq t|\mathrm{FH})}
\end{align*}
and consequently $\lambda_\alpha(u)S_\beta(u)$ is equal to the density of $T$ conditionally to $\mathrm{FH}$. Then, since $\lambda_\alpha(u)=\alpha_j$ for $u\in]c_{j-1},c_j]$ we have
\rev{\begin{align*}
\P(T_i\leq t| T_i>c_{j-1}, \mathrm{FH})& =\int_{c_{j-1}}^{t} \lambda_\alpha(u) S_\beta(u) du =\alpha_j \int_{c_{j-1}}^{t} S_\beta(u) du\\
&=\alpha_j \int_{c_{j-1}}^{t} \exp\left(-\int_0^u \lambda_{\beta}(v)dv\right) du
\end{align*}}
Now, for $u\in ]c_{j-1},t]$, $t\leq c_j$,
\begin{align*}
\int_0^u\lambda_{\beta}(v)dv=\int_0^{c_{j-1}}\lambda_{\beta}(v)dv+\beta_j (u-c_{j-1})%\sum_{k:c_k\leq c_{j-1}}\beta_k (c_k-c_{k-1})+\beta_j (u-c_{j-1})
\end{align*}
and
\begin{align*}
\int_{c_{j-1}}^{t}\exp\left(-\int_0^u \lambda_{\beta}(v)dv\right)du & =\exp\left(-\int_0^{c_{j-1}}\lambda_{\beta}(v)dv\right) \int_{c_{j-1}}^t \exp(-\beta_j (u-c_{j-1}))du\\%\exp\left(-\sum_{k:c_k\leq c_{j-1}}\beta_k (c_k-c_{k-1})\right) \int_{c_{j-1}}^t \exp(-\beta_j (u-c_{j-1}))\\
& = S_{\beta}(c_{j-1}) \int_{c_{j-1}}^t \exp(-\beta_j (u-c_{j-1}))du
\end{align*}
The integral on the right side of the equation is straightforward to compute. This gives,
\begin{align*}
S_{\beta}(c_{j-1})\int_{c_{j-1}}^t \exp(-\beta_j (u-c_{j-1}))du&=\frac{1}{\beta_j}\Big(S_{\beta}(c_{j-1})-S_{\beta}(c_{j-1})\exp(-\beta_j(t-c_{j-1}))\Big)
\end{align*}
Finally, we conclude by noticing that
\begin{align*}
S_{\beta}(t) & = \exp\left(-\int_0^{c_{j-1}} \lambda_{\beta}(u)du-\int_{c_{j-1}}^{t} \lambda_{\beta}(u)du\right)\\
&=S_{\beta}(c_{j-1})\exp\left(-\beta_j(t-c_{j-1})\right)
\end{align*}
%$\P(T_i\leq t| \mathrm{FH})=\int_{c_{j-1}}^{t} \lambda_\alpha(u) S_\beta(u) du = \frac{\alpha_j}{\beta_j} \left[ S_\beta(c_{j-1}) - S_\beta(t) \right]$
\end{proof}

\bibliographystyle{plainnat}
\bibliography{biblio}

\end{document}